\documentclass[11pt,letterpaper]{article}

\usepackage{epsfig}
\usepackage{graphicx}
\usepackage{amsthm}
\usepackage{amssymb}
\usepackage{amsmath}
\usepackage{amsfonts}
\usepackage{mathrsfs}
\usepackage{algorithm}
\usepackage{algorithmic,subfigure}

\usepackage{vmargin}
\setmarginsrb{1in}{1in}{1in}{1in}{0pt}{0pt}{0pt}{6mm}

\usepackage{color}

\newtheorem{proposition}{Proposition}
\newtheorem{lemma}{Lemma}
\newtheorem{theorem}{Theorem}

\newtheorem{definition}{Definition}

\newtheorem{claim}{Claim}

\begin{document}

\title{Outlier Detection for DNA Fragment Assembly}
\author{Christina Boucher\thanks{Department of Computer Science and Engineering, University of California, San Diego}
\and Christine Lo\footnotemark[1]
\and Daniel Lokshantov\footnotemark[1]}
\maketitle

\begin{abstract} A major impediment in the development of efficient full genome sequencing is the large portion of erroneous reads produced by sequencing platforms. Error correction is the computational process that attempts to identify and correct these mistakes. Several classical stringology problems, including the {\em Consensus String} problem, are used to model error correction. However, a significant shortcoming of using these formulations is that they do not account for a few of the reads being too erroneous to correct; these outlier strings potentially have great effect on the solution, and should be detected and removed. We formalize the problem of error correction with outlier detection by defining the {\em Consensus String with Outliers} problem. Given $n$ length-$\ell$ strings $S =\{s_1, \ldots, s_n\}$ over a constant size alphabet $\Sigma$ together with parameters $d$ and $k$, the objective in the {\em Consensus String with Outliers} problem is to find a subset $S^*$ of $S$ of size $n-k$ and a string $s$ such that $\sum_{s_i \in S^*} d(s_i, s) \leq d$. Here $d(x, y)$ denotes the Hamming distance between the two strings $x$ and $y$. We prove the following results:
\begin{itemize}
\item A variant of {\em Consensus String with Outliers} where the number of outliers $k$ is fixed and the objective is to minimize the total distance $\sum_{s_i \in S^*} d(s_i, s)$ admits a simple PTAS. Our PTAS can easily be modified to also handle the variant of the problem where a hard upper bound $d$ on the total distance is given as input, and the size of $S^*$ is to be maximized. The approximation schemes are simple enough that our results are best viewed as a performance guarantee on natural heuristics for the problem when the parameters of the heuristic are chosen appropriately.
\item Under the natural assumption that the number of outliers $k$ is small, the PTAS for the distance minimization version of {\em Consensus String with Outliers} performs well. In particular, as long as $k\leq cn$ for a fixed constant $c < 1$, the algorithm provides a $(1+\epsilon)$-approximate solution in time $f(1/\epsilon)(n\ell)^{O(1)}$ and thus, is an EPTAS. 
\item In order to improve the PTAS for {\em Consensus String with Outliers} to an EPTAS, the assumption that $k$ is small is necessary. Specifically, when $k$ is allowed to be arbitrary the {\em Consensus String with Outliers} problem does not admit an EPTAS unless FPT=W[1]. This hardness result holds even for binary alphabets.
\item The decision version of {\em Consensus String with Outliers} is fixed parameter tractable when parameterized by $\frac{d}{n-k}$. and thus, also when parameterized by just $d$.
\end{itemize}
To the best of our knowledge, {\em Consensus String with Outliers} is the first problem that admits a PTAS, and is fixed parameter tractable when parameterized by the value of the objective function but does not admit an EPTAS under plausible complexity assumptions. Hence, the proof of our hardness of approximation result combines parameterized reductions and gap preserving reductions in a novel manner. 
\end{abstract}

\newpage
\section{Introduction}

Although the laboratory methods that generate genetic sequence data have advanced remarkably since their initial use in the Human Genome Project \cite{venter}, the algorithms behind the computational methods have not advanced as dramatically.  Sajjadian et al.~\cite{alkan} describes the present time as ``watershed moment in genomics'' pointing to computational genomics as the bottleneck of the sequencing process.  In this paper, we revisit an essential problem arising in genome sequencing, reformulate this problem to better model noisy data, and show how studying approximability and parameterized complexity of this problem leads to surprising theoretical insights and algorithmic techniques that may assist in genome sequencing. 

Since the discovery of DNA as the basic unit of heredity, significant effort has been focused on automated determination of the sequence of nucleotides corresponding to a sample of DNA, a process referred to as {\em genome sequencing}. The key technology this process relies on is the sequencing platform that accepts a collection of biological (DNA) samples and produces {\em reads} from the samples.  A read is a string from the alphabet $\{$\texttt{A}, \texttt{C}, \texttt{G}, \texttt{T}$\}$ that represents the sequence of nucleotides in a sample. Sequencing platforms are extremely limited in that they cannot process the entire DNA sample at once but rather, they handle very small pieces of the DNA at a time.  The resulting problem for an average-size genome of length 4 million is that $\sim$20 million reads of length 50 must be assembled into one contiguous piece.  This computational process of building the contiguous string from reads is referred to as {\em fragment assembly}, and is especially challenging--if not, impossible--for complex genomes with higher repeat and duplication content.   While the current generation of sequencing platforms can produce a large amount of reads in a relatively short period of time, the reads they produce are greatly error prone, increasing the computational difficulty of fragment assembly.  Error correction, which is vital in genome assembly, aims to identify and correct any mistakes made by the sequencing platform and thus, reduces the computational demands of the fragment assembly algorithms \cite{PTW01}.  

Contamination of the DNA sample and erroneous runs of the sequencing platforms are frequent occurrences that lead to many reads having a large fraction of errors and hence, deviate quite dramatically from the rest of the data.  Ideally, these ``outlier'' strings should be detected and removed from the input prior to assembly.  Although, problem formulations with outliers have been previously proposed and studied in different contexts--including machine learning \cite{CKMN01,PW04,HCKKF05}, network design problems \cite{AP08,AGLS10,CKMN01,C08,G05,GW95}, and bioinformatics \cite{BM11,lanctot}--it has not been considered or proposed in error correction of genome sequencing data.  Present error correction methods do not account for the possibility of outliers, and hence,  are required to be highly liberal in the elimination of data. Therefore, they remove a large number of reads that could have been used for assembly.  We introduce the following formulation of error correction that also captures the existence of outliers in the data. \\

\noindent{\em Consensus String with Outliers} \\
\noindent{\em Input:} a set\footnote{Technically, this is a multi-set since we allow any string to occur multiple times.} of $n$ length-$\ell$ strings $S = \{s_1, \ldots, s_n\}$ over a finite alphabet $\Sigma$ and nonnegative integers $k$ and $d$.\\
\noindent{\em Question:} Find a length-$\ell$ string $s$ and subset $S^*$ of $S$ of size $n-k$, where $\sum_{\forall s_i \in S^*} d(s, s_i)\leq d$.\\

We restrict interest to Hamming distance and denote $d(x, y)$ to be the Hamming distance between the length-$\ell$ strings $x$ and $y$. The following are natural optimization versions of {\em Consensus String with Outliers} that we will consider:  
\begin{itemize}
\item {\em Consensus String with Max Non-Outliers}: given $n$ length-$\ell$ strings $S = \{s_1, \ldots, s_n\}$ over a finite alphabet $\Sigma$ and nonnegative integer $d$, the aim is to find a consensus string $s$ and subset of $S^* \subset S$, where $|S^*|$ is maximal and $\sum_{\forall t \in S^*} d(s,t)\leq d$. 
\item {\em Min-distance Consensus String with Outliers}: given $n$ length-$\ell$ strings $S = \{s_1, \ldots, s_n\}$ over a finite alphabet $\Sigma$ and nonnegative integer $k$, the aim is to find a consensus string $s$ and subset of $S^* \subset S$, where $n - |S^*| = k$ and $\sum_{\forall t \in S^*} d(s,t)$ is minimal.
\end{itemize}
\smallskip

\paragraph{Our Results.} The problems considered are NP-hard in general, however, they turn out to be amenable to approximation and parameterized algorithms.  A {\em polynomial-time approximation scheme} (PTAS) for a minimization problem is an algorithm which takes an instance of the problem and a parameter $\epsilon > 0$ and, in polynomial time, produces a solution that is within a factor $1 + \epsilon$ of being optimal. If the exponent of the polynomial in the running time of the algorithm is independent of $\epsilon$ then the PTAS is said to be an {\em efficient PTAS} (EPTAS). We present several results on the ability to efficiently solve and approximate the above optimization problems within arbitrarily small factors, and demonstrate the tightness of these results. Specifically, we prove the following: 
\begin{itemize}
\item There exists a deterministic PTAS for {\em Min-distance Consensus String with Outliers} and {\em Consensus String with Max Non-Outliers}.  
\item For instances where $k < cn$ and fixed $c < 1$, the PTAS for {\em Min-distance Consensus String with Outliers} can be improved to a randomized EPTAS. 
\item In the general case, both {\em Min-distance Consensus String with Outliers} and {\em Consensus String with Max Non-Outliers} do not admit an EPTAS, unless FPT=W[1]. Thus, the requirement that $k < cn$ is necessary to improve the PTAS for {\em Min-distance Consensus String with Outliers} to an EPTAS.
\item {\em Consensus String with Outliers} can be solved to in time $\delta^{O(\delta)}|\Sigma|^{\delta} n^9$, where $\delta = d/(n  - k)$. 
\end{itemize}
For a parameter $\delta$, an algorithm with running time $f(\delta)n^{O(1)}$ is called a {\em fixed parameter tractable} (FPT) algorithm for the problem parameterized by $\delta$. Parameterized problems that admit such algorithms are said to be FPT. Hence our algorithm for {\em Consensus String with Outliers} proves that the problem is FPT parameterized by $\delta$. 

Our approximation schemes are based on random sampling. If the number of outliers is small, then with reasonably high probability a small random subset of the input strings will not contain any outliers. If the random sample does not contain outliers then the sample can be used to estimate the optimal consensus string.  We show that if the size of the sample and the number of repetitions of the experiment are chosen appropriately then there exists a good bound on the quality of the output of this natural heuristic.  For inputs where the noise does not completely overwhelm the data, i.e.~when $k \leq cn$ for $c < 1$, the dependence on the running time of our approximation scheme for {\em Min-distance Consensus String with Outliers} is good; more specifically, it is an EPTAS. 

The difference in running time of a PTAS and an EPTAS can be quite dramatic. For instance, running a $O(2^{1/\epsilon}n)$-time algorithm is reasonable for $\epsilon=\frac{1}{10}$ and $n=1000$, whereas running a $O(n^{1/\epsilon})$-time algorithm is infeasible. Hence, considerable effort has been devoted to improving PTASs to EPTASs, and showing that such an improvement is unlikely for some problems.  For example, Arora~\cite{Arora96} gave a $n^{O(1/\epsilon)}$-time PTAS for {\em Euclidean TSP}, which was then improve to a $O(2^{O(1/\epsilon^2)}n^2)$-time algorithm in the journal version of the paper~\cite{Arora98}.  On the other hand {\em Independent Set} admits a PTAS on unit disk graphs~\cite{HuntMRRRS98} but Marx~\cite{marx05} showed that, unless FPT=W[1], it does not admit an EPTAS. Many more examples of PTASs that have been improved to EPTASs, and problems for which a PTAS exists but for which an EPTAS has been ruled out under the assumption that FPT$\neq$W[1] can be found in the survey of Marx~\cite{marx_survey}. An interesting question is whether the requirement that $k \leq cn$ for $c < 1$ is necessary in order to improve the PTAS for {\em Min-distance Consensus String with Outliers} to an EPTAS. Can an EPTAS be obtained for this problem without the requirement? 


A useful observation in this regard is that an EPTAS for an optimization problem automatically yields a FPT algorithm for the corresponding decision problem parameterized by the value of the objective function~\cite{marx_survey}. More specifically, if we set $\epsilon = \frac{1}{2\alpha}$, where $\alpha$ is the value of the objective function, then a $(1+\epsilon)$-approximation algorithm would distinguish between ``yes'' and ``no'' instances of the problem. Hence, an EPTAS could be used to solve the problem in $O(f(\epsilon)n^{O(1)}) = O(g(\alpha)n^{O(1)})$-time. This observation is frequently used to rule out the existence of an EPTAS. If a problem does not admit a FPT algorithm parameterized by the value of the objective function unless FPT=W[1], then the corresponding optimization problem does not admit an EPTAS unless FPT=W[1].

To the best of our knowledge {\em all} known results ruling out EPTASs for problems for which a PTAS is known use this approach. Unfortunately, it cannot be used to rule out an EPTAS for {\em Min-distance Consensus String with Outliers} because  {\em Consensus String with Outliers} parameterized by $d$ is FPT. In particular, we show there is an algorithm for {\em Consensus String with Outliers} with running time $\delta^{O(\delta)}|\Sigma|^{\delta} n^9$, where $\delta = d/(n  - k)$, and since $\delta$ is always at most $d$--and much smaller than $d$ for most inputs--this algorithm runs in $O(d^{O(d)}|\Sigma|^{d} n^9)$-time. Our FPT algorithm is an adaptation of the algorithm by Marx~\cite{M08} for the {\em Consensus Patterns} problem.

In his survey, Marx~\cite{marx_survey} introduces a hybrid of FPT reductions and gap preserving reductions and argues that it is conceivable that such reductions could be used to prove that a problem that has a PTAS and is FPT parameterized by the value of the objective function does not admit an EPTAS unless FPT=W[1]. We show that {\em Min-distance Consensus String with Outliers} does not admit an EPTAS unless FPT=W[1], giving the first example of this phenomenon.  At the core of our reduction is an analysis of one-dimensional random walks where some of the steps are ``double steps'' that are taken in the same direction. The results on random walks could turn out useful in other hardness proofs, and thus, might be of independent interest. Parameterized hardness results for a few other parameterizations of {\em Consensus String with Outliers} follow as simple corollaries of our construction.

\subsection*{Related Work} The problems considered in this paper belong to the more general class of stringology problems where a set of strings is given and the aim is to determine a single string that is representative for the set. The exact definition of what being a good representative means may vary and different definitions lead to abstractions of various problems in bioinformatics \cite{lanctot}. The {\em Consensus Patterns} problem is quite similar to our problem, however, in this context the aim is to find a substring in each of the input strings and consensus string so that the sum of the Hamming distances is minimized.  Li et al.~\cite{LMW022} gave a PTAS for this problem and there has been a significant effort in attempting on proving tighter bounds on the running time of the PTAS \cite{brona,brona2}.  The {\em Closest String} problem is another related problem where the goal is to find a string that minimizes the maximum Hamming distance to any string. This problem also admits a PTAS but no EPTAS \cite{LMW02}.  Both problems have been investigated in the framework of parameterized complexity by several authors, however, the parameterization of {\em Consensus Patterns} with respect to the distance appeared to be very challenging.  In 2005, Marx \cite{M08} showed {\em Consensus Patterns} is FPT when parameterized by $\delta = d/n$ and bounded alphabet size.  

\subsection*{Overview of DNA Fragment Assembly}

The general approach to large-scale sequencing is as follows: first the DNA is extracted from the cell and copied multiple times, then the DNA is cut into smaller fragments, each fragment is sequenced by a sequencing platform to produce a read, and finally the reads are assembled into large segments of the genome.  Figure \ref{fig:error_rate} illustrates this process. One important point is that sequencing platforms can produce many hundreds of thousands, or even millions of reads in a short time (on the order of a day), but can only handle small segments of DNA at a time and produce relatively short reads.  The copying step ensures that a position of the genome is sequenced multiple times and the reads overlap by an adequate amount.  This overlap is what allows for the assembly of the reads into large contiguous strings. Developing novel algorithms and tools for the fragment assembly process is, at present, a very active area of research in bioinformatics. Current assembly tools are efficient, however, their accuracy is substantially diminished by repeated regions in the genome sequence and sequencing errors.

As previously mentioned, error correction of the reads is an important step in genome sequencing, however, present algorithms are still unable to handle outliers in the data.  The majority of sequencing errors occur when the nucleotide found in a read deviates from the actual nucleotide in the DNA sample ({\em i.e.}~a read has the symbol \texttt{A} at a position where it should be a \texttt{C}), making Hamming distance the most reasonable metric to use.  In a read, for the first 50 positions the error rate is quite small but for subsequent positions the error probability increases exponentially \cite{altacyclic,SS11}.  See Figure \ref{fig:error_rate}. This is why the length of the reads is at most $70$-$100$.  Due to this change in the error probability and technical details related to fragment assembly\footnote{We leave out these details in this paper and direct interested readers to the work of Pevzner et al.~\cite{PTW01}}, error correction begins by computing the set of all consecutive, length-$\ell$ substrings from each read, where $\ell$ is an input parameter.  Hence, error correction is implicitly performed on the set for all length-$\ell$ contiguous substrings of reads rather than the full reads. 

The majority of error correction algorithms consider the first 50 positions in a read and ignore remaining positions \cite{KCS11,PTW01,YAD11}, eliminating a large portion of the data.  This is unsatisfactory since acquiring the data is both expensive and time consuming, and any loss of data will affect the accuracy of the assembly.  Due to the change in the error probability in the reads, some of the length-$\ell$ strings will have a significant but tolerable number of errors (i.e.~up to 15\% of positions being erroneous) that can be error corrected and thus, used in fragment assembly. On the other hand, the length-$\ell$ strings that stem from contaminated data or bad runs of the sequencing platform should be detected and removed.

\begin{figure}[t]
\centering
\includegraphics[width=\textwidth]{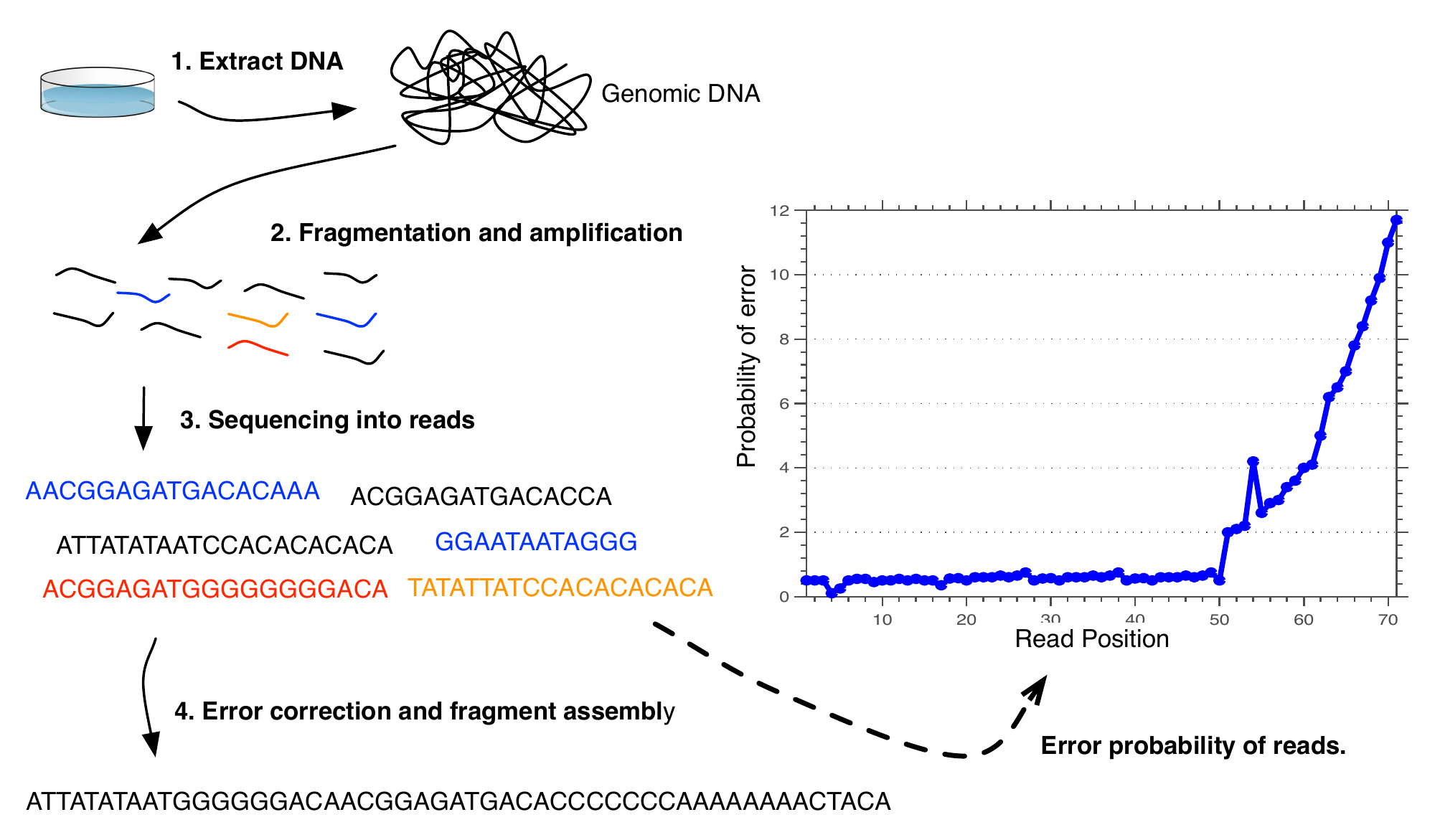}
\caption{A visualization of the basic steps needed for whole genome assembly of a biological sample.  The probability that a character in a read was sequenced incorrectly is highly dependent on its position within the read. The change in the error probability with respect to the read length is illustrated \cite{altacyclic}.}
\label{fig:error_rate}
\end{figure}

\subsubsection*{Preliminaries}

A maximization problem admits a PTAS if there is an algorithm $A(\mathcal{I},  \epsilon)$  such that, for any $\epsilon > 0$ and any instance $\mathcal{I}$ of $A(\mathcal{I}, \epsilon)$ outputs a $(1 - \epsilon)$-approximate solution in time $|\mathcal{I}|^{f(1/\epsilon)}$ for some function $f$. A PTAS for a minimization problem finds a $(1 + \epsilon)$-approximate solution in time $|\mathcal{I}|^{f(1/\epsilon)}$.  An approximation scheme where the exponent of $|\mathcal{I}|$ in the running time is independent of $\epsilon$ is called an {\em efficient} polynomial time approximation scheme (EPTAS). Formally, an EPTAS is a PTAS whose running time is $f(1 / \epsilon)^{O(1)} |\mathcal{I}|^{O(1)}$.  

We give a brief introduction to paramterized complexity. A problem $\varphi$ is said to be {\em fixed parameter tractable} with respect to parameter $k$ if there exists an algorithm that solves $\varphi$  in $f(k) \cdot n^{O(1)}$ time, where $f$ is a function of $k$ that is independent of $n$ \cite{DF99}.  The class of all fixed parameter tractable problems is denoted by FPT. The class W[1] of parameterized problems is the basic class for fixed parameter intractability, FPT $\subseteq$ W[1] and the containment is believed to be proper. A parameterized problem $\Pi$ with the property that an FPT algorithm for $\Pi$ would imply that FPT=W[1] is called W[1]-hard. Downey and Fellows \cite{DF99} define {\em fpt-reductions}, which preserve W[1]-hardness.
%


%
%
Let $L, L' \subseteq \sum^* \times \mathbb{N}$ be two parameterized problems. We say that $L$ {\em fpt-reduces} to $L'$ if there are functions $f, g : \mathbb{N} \rightarrow \mathbb{N}$ and an algorithm that given an instance $(\mathcal{I},k)$ runs in time $f(k)|\mathcal{I}|^{f(k)}$ and outputs an instance $(\mathcal{I'},k')$ such that $k' \leq g(k)$ and $(\mathcal{I},k) \in L \iff (\mathcal{I}',k') \in L'$.
%
These reductions work as expected; if $L$ fpt-reduces to $L'$ and $L'$ is FPT then so is $L'$. Furthermore, if $L$ fpt-reduces to $L'$ and $L$ is W[1]-hard then so is $L'$. We refer the reader to the textbooks~\cite{DF99,niedermeier,flum} for a more thorough discussion of parameterized complexity.

Let $s$ be a string over the alphabet $\Sigma$. We denote the length of $s$ as $|s|$, and the $j$th character of $s$ as $s[j]$.  Hence, $s = s[1]s[2]\ldots s[|s|]$. For a set $S$ of strings of the same length we denote by $S[i]$ as $\{s[i]~:~s \in S\}$. That is, if the same character appears at position $i$ in several strings it is counted several times in $S[i]$. For an interval $P=\{i,i+1,\ldots,j-1,j\}$ of integers, define $s[P]$ to be the substring $s[i]s[i+1]\ldots s[j]$ of $s$. For a set $S$ of strings and interval $P$ define $S[P]$ to be the (multi)set $\{s[P]~:~s \in S\}$. For a set $S$ of length-$\ell$ strings the consensus string of $S$, denoted as $c(S)$, is such that $c(S)[i]$ is the most-frequent character in $S[i]$ for all $i \leq \ell$. Ties are broken by selecting the lexicographically first such character, however, we note that the tie-breaking will not affect our arguments. 

We denote the sum Hamming distance between a string, $s$, and a set of strings, $S$, as $d(S, s)$. Observe that the consensus string $c(S)$ minimizes $d(S, c(S))$--that is no other string $x$ is closer to $S$ than $c(S)$. However, some $x \neq c(S)$ could achieve $d(S, x)=d(S, c(S))$ and we refer to such strings as {\em majority strings} because they are obtained by picking a most-frequent character at every position with ties broken arbitrarily. The {\em Consensus String With Outliers} problem can now be succinctly stated as follows: given a set $S$ of strings and integers $k$ and $d$, the objective is to find a subset $S^* \subseteq S$ of size $n^*=n-k$ such that $d(S, c(S)) \leq d$, if it exists.

Given a subset $S^* \subseteq S$ we can compute $c(S^*)$ in polynomial time by choosing a majority string for $c(S^*)$. If we are given $c(S^*)$ for the optimal solution $S^*$ (but not given $S^*$ itself) then we can recover $S^*$ from $c(S^*)$ and $S$ in polynomial-time since $S^*$ is the $n-k$ strings in $S$ that are closest to $c(S^*)$. Similarly, given any string $x$, we denote $S_x$ as the subset of $S$ containing the $n^*$ strings closest to $x$. By construction $S_x$ satisfies the following inequality: $d(S', x) \geq d(S_x, x) \geq d(S_x, c(S_x))$ for any subset $S'$ of $S$ of size $n^*$.

\section{Approximating Consensus String with Outliers} \label{sec:ptas}

We prove the existence of a PTAS for the {\em Min-distance Consensus String with Outliers} problem. Our algorithm is based on random sampling. For a given value of $\epsilon$, the algorithm selects a value for the parameter $r$ based on $\epsilon$, picks $r$ strings $S' = (s'_1, s'_2, ... s'_r)$ from $S$ uniformly at random (with replacement), and returns the consensus string corresponding to $S'$. The next lemma shows that if $S'$ was taken from a (unknown) optimal solution $S^*$, rather than from the entire input set $S$, then in expectation $c(S')$ is almost as good the consensus string for the set $S^*$.

Our arguments rely on well-known concentration bounds for sums of independent random variables. We use the following variant of the Hoeffding's bound~\cite{H63} given by Grimmett and Stirzaker \cite[p. 476]{GR}.  

\begin{proposition}\label{prop:hoeff} {\bf (Hoeffding's bound) } Let $X_1, X_2, ... X_n$ be independent random variables such that $a_i \leq X_i \leq b_i$ for all $i$. Let $X = \Sigma_i X_i$ and the expected value of $X$ be $E[X]$ then it follows that: $$\Pr[X-E[X] \geq t] \leq \exp\left(\frac{-2t^2}{\Sigma_{i=1}^n \left(b_i-a_i\right)^2}\right).$$ 
\end{proposition}

\begin{lemma}\label{lem:min_d} For all $\epsilon > 0$ and $\sigma$, there exists a value of $r$ such that the following holds: if $S$ is a set of length-$\ell$ strings over the alphabet $\Sigma$, where $|\Sigma| = \sigma$, and $S'$ is a subset of $S$ of size $r$, $(s'_1, s'_2, ... s'_r)$, chosen uniformly at random, then $E[d(S, c(S'))] \leq (1+ \epsilon)d(S, c(S))$. \end{lemma}

\begin{proof} 
We prove that there exists a $r$ such that $E[d(S, c(S'))] \leq (1+ 2\epsilon)d(S, c(S))$. Applying this weaker inequality with $\epsilon' = \epsilon/2$ then proves the statement of the Lemma. We assume, without loss of generality, that $c(S)$ is equal to $0^{\ell}$, $\epsilon \leq 1/16$, and $r \geq 8$. We restrict interest to column $i$ of $S$, where $0 \leq i \leq \ell$, let $d_i$ be the number of nonzero symbols in column $i$ and let $z_i = n-d_i$. Observe that $d(S, c(S'))$ is equal to the sum over $i$ of the number of strings $s \in S$ such that $s[i] \neq c(S')[i]$. By linearity of expectation it is sufficient to prove that for every $i$ we have  $E[d(S[i], c(S')[i])] \leq (1+ 2\epsilon)d_i$.

First, we assume $d_i$ is at most $\epsilon n$. Let $q$ be the probability that $c(S')[i] \neq 0$. It follows that $E[d(S[i], c(S')[i])]$ is at most $d_i(1-q) + qn$. We determine an upper bound on the probability $q$ as follows: $$ q \leq \sum_{x = \lceil r / 2 \rceil}^r {r \choose x} \left( d_i / n \right)^x \left(1 - d_i / n \right)^{r - x} \leq  \sum_{x = \lceil r / 2 \rceil}^r 2^r \left( d_i / n \right)^x \leq  2^r \left( d_i / n \right)^{\lceil r / 2 \rceil} \frac{1- \left( d_i / n \right)^{\lceil r / 2 \rceil}}{1 - (d_i/n)}.$$  
Since $d_i / n \leq \epsilon \leq 1/16$, we get: $$q \leq 2^{r + 1}  \left(d_i / n \right)^{\lceil r / 2 \rceil} \leq 2^{r + 1}  \epsilon^{\lceil r /4 \rceil} \left( d_i / n \right)^{\lfloor r / 4 \rfloor} \leq 2^r \left( \frac{1}{16}\right)^{\lceil r / 4 \rceil} \cdot 2 \left( d_i / n \right)^{\lfloor r / 4 \rfloor} = 2 \left( d_i / n \right)^{\lfloor r / 4 \rfloor}.$$ It follows from the last inequality, and that $r \geq 8$, that $q \leq 2 \left( d_i / n \right)^2$. Hence, we obtain the following bound on $E[d(S[i], c(S')[i])]$:

$$E[d(S[i], c(S')[i])]  \leq d_i(1 - q) + qn   
					   		\leq d_i + 2 \left( \frac{d_i}{n}\right)^2 n  
							\leq (1 + 2 \epsilon)d_i $$

Next, we assume that  $d_i > \epsilon n$.  We say that a symbol $\alpha \in \Sigma$ is a {\em good} symbol if there are at least $z_i-n\epsilon^2$ strings in $S$ that have the symbol $\alpha$ at column $i$; any symbol that is not good is {\em bad}. If $c(S')[i]$ is a good symbol then $d(S[i], c(S')[i])$ is at most $d_i + n\epsilon^2$ and hence, is at most $(1+\epsilon)d_i$ since $d_i > \epsilon n$. 
Let $p$ be the probability that $c(S')[i]$ is a bad symbol then, $E[d(S[i], c(S')[i])]$ is upper bounded by $(1-p)(1+\epsilon)d_i + pn$. Lastly, we determine an upper bound on $p$ to complete the proof.

Let $\alpha$ be a bad symbol and $p_\alpha$ be the probability that $c(S')[i]$ is equal to $\alpha$. We note that in order for $c(S')[i]$ to be $\alpha$, there has to be more positions equal to $\alpha$ than $0$ in $S'[i]$.  Let $X$ be the difference between the number of positions equal to $\alpha$ and the number of positions equal to $0$ in $S'[i]$.  It follows that $p_\alpha \leq \Pr[X \geq 0]$.
Let $X_j$ be an indicator variable which is $1$ if $s'_j[i]$ is equal to $\alpha$, -1 if it is equal to $0$, and $0$ otherwise.  Since $\alpha$ is a bad symbol, there are at least $\epsilon^2$ more positions equal to $0$ than positions equal to $\alpha$ in $S'[i]$ and therefore, $E[X_j] = \Pr[s'_j[i] = 0] - \Pr[s'_j[i] = \alpha ] \leq - \epsilon^2$. By linearity of expectation, we obtain $E[X] = \Sigma_{j=1}^r E[X_j] \leq - r\epsilon^2$. Using this inequality, we get $\Pr[X \geq 0] \leq \Pr[X - E[X] \geq r\epsilon^2]$. Since the $X_j$ variables are independent and difference between the upper and lower bound of $X_j$ is $2$, we can use Hoeffding's inequality to obtain the following bound.
$$\Pr[X-E[X] \geq r\epsilon^2] \leq \exp \left( \frac{-2 r^2 \epsilon^4}{r 2^2} \right)= \exp \left( \frac{r \epsilon^4}{2}\right)$$
By choosing $r= \max\left(\frac{2\ln(\frac{\sigma}{\epsilon^2})}{\epsilon^4}, 8\right)$, we get $p_\alpha \leq \frac{\epsilon^2}{\sigma}$. 
Finally, we bound $p$ as follows: $p \leq \sum\limits_\alpha p_\alpha \leq \sigma \frac{\epsilon^2}{\sigma} = \epsilon^2$. We can now use the upper bound on $p$ and our assumption that $d_i > \epsilon n$ to bound $E[d(S[i], c(S')[i])]$:
$$ E[d(S[i], c(S')[i])] 	  \leq (1-p)(1+\epsilon)d_i + pn  \leq (1+ \epsilon)d_i + \epsilon^2n  \leq (1+ 2\epsilon)d_i.$$
This concludes the proof. 
\end{proof}

Lemma~\ref{lem:min_d} gives a simple, deterministic PTAS for {\em  Min-distance Consensus String with Outliers}.

\begin{theorem} \label{lem:min_d_det} There exists a PTAS for Min-distance Consensus String with Outliers. \end{theorem}
\begin{proof}
It follows from Lemma~\ref{lem:min_d} that there exists an integer $r$ such that $E[d(S^*, c(S'))] \leq (1+\epsilon)d(S^*, c(S^*))$ if $S'$, the set of $r$ of strings chosen from $S$, is from an (unknown) optimal solution $S^*$. Some subset $S'$ of $S^*$ must achieve expectation. The algorithm guesses this set $S'$ by trying all possible $n^r$ subset of $S$ of size $r$. Let $x=c(S')$. The algorithm returns the set $S_x$ of the $n^*$ strings closest to $x$. This set satisfies $d(S_x, c(S_x)) \leq d(S_x, x) \leq d(S^*, x) \leq (1+\epsilon)d(S^*,c(S^*))$, concluding the proof.  
\end{proof}

If the number of outliers $k$ is small compared to $n$, i.e.~$k \leq n/2$, then with probability $1/2^r$ a random subset $S'$ of $r$ strings is a subset of an optimal solution $S^*$. We use this to give a randomized EPTAS for {\em Min-distance Consensus String with Outliers}.


\begin{theorem} \label{lem:min_d_rand} There exists a randomized EPTAS for Min-distance Consensus String with Outliers for inputs when $k \leq cn$ for $c < 1$. The algorithm runs in time $\frac{1}{(1-c)^r} \cdot f(\epsilon)(n\ell)^{O(1)}$ and outputs a $(1+\epsilon)$-approximate solution with probability $1/2$.
\end{theorem}

\begin{proof} 
We give a polynomial-time algorithm that returns a $(1+\epsilon)$-approximate solution with probability $(1-c)^r \cdot f(\epsilon)$. Repeating this algorithm $O\left(\frac{1}{(1-c)^r \cdot f(\epsilon)} \right)$ times then yields the statement of the theorem. The algorithm selects a value for $r$ such that for a random subset $S'$ of the unknown optimal solution $S^*$ the inequality $E[d(S^*, c(S'))] \leq (1+\frac{\epsilon}{3})d(S^*, c(S^*))$ holds.  It follows from Lemma~\ref{lem:min_d} that this can be done so that $r$ only depends on $\epsilon$.  Next, $r$ strings from $S$ are selected uniformly at random (with replacement) to form a subset $S'$. Let $x = c(S')$. The algorithm then returns the set $S_x$ of the $n^*$ strings closest to $x$.

It remains to find a sufficient lower bound of the probability that the returned set is a $(1+\epsilon)$-approximation. Since $k \leq cn$, it follows that the probability that $S'$ is taken from an (unknown) optimal solution $S^*$ is at least $(\frac{n-cn}{n})^r = (1-c)^r$.  If $S'$ is taken from $S^*$ then by Lemma~\ref{lem:min_d} we have that $E[d(S^*, c(S'))] \leq (1+\frac{\epsilon}{3})d(S^*, c(S^*))$. Next, we assume otherwise.  By Markov's inequality \cite[p. 311]{GR} the probability that $d(S^*, c(S'))$ exceeds expectation by a factor at least $1+\frac{\epsilon}{3}$ is at most $\frac{1}{1+\frac{\epsilon}{3}}$. Hence, with probability $f(\epsilon)$ for some function $f$ of $\epsilon$ we have that: $$d(S^*, c(S')) \leq \left( 1+\frac{\epsilon}{3} \right) d(S^*, c(S^*)) \cdot \left( 1+\frac{\epsilon}{3} \right),$$ which is at most $(1+\epsilon)d(S^*, c(S^*))$ when $2\left( \frac{\epsilon}{3} \right)^2 \leq \frac{\epsilon}{3}$. In particular, this holds if $\epsilon \leq 1/3$, concluding the proof.
\end{proof}

The best way to view Theorem~\ref{lem:min_d_rand} is as a performance guarantee on a natural heuristic for the problem when the parameter $r$ is chosen appropriately. We note that one would expect natural inputs to contain substantially fewer outliers than $n/2$, and that Markov's inequality is a very pessimistic bound for the probability of achieving expectation. Hence, it is likely that for reasonable inputs the above algorithm will perform much better in practise than the proved bounds. 

We now show that the PTAS for the {\em Min-distance Consensus String with Outliers} problem can be extended to obtain a PTAS for the {\em Consensus String with Max Non-Outliers} problem.  We recall that for this optimization problem, we are given $S$ and an integer $d$ and asked to find a set $S^* \subseteq S$ that maximizes $|S^*|$ and satisfies the constraint $d(S^*, c(S^*)) \leq d$.

\begin{theorem} \label{thm:max_n_star} There exists a PTAS for Consensus String with Max Non-Outliers.\end{theorem}

\begin{proof} 
We give a $(1-2\epsilon)$-approximation algorithm that runs in $O((n\ell)^{f(\epsilon)})$-time. We denote an (unknown) optimal solution as $S^*$, and let $n^* = |S^*|$. A subset $S' \subseteq S$ is said to be {\em feasible} if $d(S', c(S')) \leq d$. First, the algorithm enumerates all subsets of $S$ of size at most $1/\epsilon$ and keeps the largest feasible set. Next, the algorithm guesses $n^*$ (by trying all possibilities) and applies the algorithm from Theorem~\ref{lem:min_d_det} to find a set $S_x$ of size $n^*$ and a string $x$ such that $d(S_x, x) \leq (1+\epsilon)d(S^*, c(S^*)) \leq (1+\epsilon)d$. It then constructs $S''$ by removing the $\lceil \epsilon n^* \rceil$ strings furthest away from $x$ from $S_x$. Since 
$$d(S'', c(S'')) \leq d(S'', x) \leq  (1-\epsilon)  d(S_x, x) \leq (1-\epsilon)(1+\epsilon)d \leq d,$$ 
it follows that $S''$ is feasible. The algorithm returns either $S''$ or the largest feasible set found in the first phase, which ever is largest. The running time is clearly bounded by $O\left((n\ell)^{f(\epsilon)} \right)$, so it remains to prove that the returned set is in fact a $(1-\epsilon)$-approximation of $S^*$. If $n^* < \frac{1}{\epsilon}$ the algorithm finds and returns $S^*$. Next, if $n^* \geq \frac{1}{\epsilon}$ it follows that $|S''| \geq |S'|-\lceil \epsilon n^* \rceil \geq n^*(1-\epsilon)-1 \geq n^*(1-2\epsilon)$, concluding the proof. 
\end{proof}
\section{Hardness Results}\label{sec:approx}
For reasonable instances of the {\em Min-distance Consensus String With Outliers} problem, we expect the number of non-outliers to be greater than the number of outliers.  As we have seen, Theorem~\ref{lem:min_d_rand} gives an EPTAS for most instances of the {\em Min-distance Consensus String With Outliers} problem--namely those where $k \leq cn$ for $c < 1$. When $k$ cannot be upper bounded in this manner then the noise is much stronger than the signal, and there is little hope for accurate error correction.  Further, {\em Min-distance Consensus String with Outliers} should not be seen as an error correction problem when $k$ is almost equal to $n$, but rather the problem of finding the ``densest possible'' cluster of points in Hamming space.  Determining whether the requirement that $k \leq cn$ for $c < 1$ is necessary in order to improve the PTAS from Theorem~\ref{lem:min_d_det} to an EPTAS warrants further investigation.  Now, we prove this requirement is unavoidable since the general version of {\em Min-distance Consensus String with Outliers} does not admit an EPTAS unless FPT=W[1].

\begin{theorem}\label{thm_no_eptas} There exists no EPTAS for {\em Min-distance Consensus String With Outliers}, unless FPT = W[1]. \end{theorem}

The proof of Theorem~\ref{thm_no_eptas} is by reduction from the {\em MultiColored Clique (MCC)} problem. Here input is a graph $G$, an integer $k$ and a partition of $V(G)$ into $V_1 \uplus V_2 \ldots V_k$ such that for each $i$, $G[V_i]$ is an independent set. The task is to determine whether $G$ contains a clique $C$ of size $k$. Observe that such a clique must contain exactly one vertex from each $V_i$, since for each $i$ we have $C \cap V_i \leq 1$. It is known that MCC cannot be solved in time $f(k)n^{O(1)}$, unless FPT=W[1]~\cite{FellowsHRV09}.

Given an instance $(G,k)$ of MCC we produce in time $f(k)n^{O(1)}$ an instance $(S,n^*)$ of {\em Min-distance Consensus String with Outliers} with the following property. If $G$ has a $k$-clique then there exists an $S' \subset S$ of size $n^*$ such that $d(S', c(S')) \leq D_{yes}$, whereas if no $k$-clique exists in $G$ then for each $S' \subset S$ of size $n^*$ we have $d(S', c(S')) \geq D_{no}$. The values of $D_{yes}$ and $D_{no}$ will be chosen later in the proof, but the crux of the construction is that $D_{no} \geq \left( 1+\frac{1}{h(k)} \right)D_{yes}$. Hence, one could use the reduction together with an EPTAS for {\em Min-distance Consensus String with Outliers} setting $\epsilon = \frac{1}{2h(k)}$ to solve the MCC problem in time $g(k)n^{O(1)}$. This reduction is a parameterized, gap-creating reduction where the size of gap decreases as $k$ increases but the decrease is a function of $k$ only.

\smallskip
\paragraph{Construction.} We describe how the instance $(S,n^*)$ is constructed from $(G,k)$. Our construction is randomized, and will succeed with probability $\frac{2}{3}$. To prove Theorem~\ref{thm_no_eptas} we have to change the construction to make it deterministic but for now let us not worry about that. We start by considering the instance $(G,k)$ and let $E(G) = \{e_1, e_2, \ldots e_m\}$. We partition the edge set $E(G)$ into sets $E_{p, q}$ where $1 \leq p < q \leq k$ as follows; $e_i \in E_{p, q}$ if $e_i = uv$, $u \in V_p$ and $v \in V_q$. 

Edges of $G$ are unordered pairs $uv$ of vertices of $G$. An {\em edge endpoint} $\hat{e}$ is an {\em ordered} pair $(u,v)$ of vertices of $G$ such that $uv$ is an edge of $G$. We denote the set of all edge endpoints of $G$ by $\hat{E}(G) = \{\hat{e}_1, \hat{e}_2, \ldots \hat{e}_{2m}\}$. There are two edge endpoints that correspond to the same edge. For two edge endpoints $\hat{e}_p$ and $\hat{e}_q$ that both correspond to the edge $e_r$ we say that $\hat{e}_p \sim \hat{e}_q$, $\hat{e}_p \sim e_r$ and that $\hat{e}_q \sim e_r$. For every $i \leq k$ define the set $\hat{E}_i = \{(u,v) \in \hat{E}~:~u \in V_i\}$.

Based on $G$ and $k$, we select two integers $\ell_1$ and $\ell_2$, that satisfy the following proerties; $\ell_1 = f \cdot \log n$, $\ell_2 = g \cdot \ell_1$ for some $f \geq 1$ and $g \geq 1$ that depend only on $k$. The exact value of $\ell_1$ and $\ell_2$ will be discussed later in the proof. We construct a set $Z = {z_1, z_2, \dots z_{2m}}$ of strings, $Z$ will act as a ``pool of random bits'' in our construction. For each endpoint $\hat{e}_i \in \hat{E}(G)$ we make a string $z_i$ as follows. 
$$z_i = \overline{a}_i^1 \circ \overline{a}_i^2 \ldots \circ \overline{a}_i^k \circ \overline{b}_i^{1,2} \circ \overline{b}_i^{1,3} \ldots \overline{b}_i^{1,k} \circ \overline{b}_i^{2,3} \circ \overline{b}_i^{2,4} \ldots \circ \overline{b}_i^{k-1,k}$$
For every $p$, $\overline{a}_i^p$ is a random binary string of length $\ell_1$. For every $p$ and $q$, $\overline{b}_i^{p, q}$ is a random binary string of length $\ell_2$. For each $p$ and vertex $u \in V_p$ we make an identification string $id(u)$ of length $\ell_1$. Let $i$ be the smallest integer such that the edge endpoint $\hat{e}_i$ is $(u,v)$ for some $v$. We set $id(u) = \overline{a}_i^p$. Similarly, for every pair of integers $p \leq q$ and each edge $e \in E_{p,q}$ make an identification string $id(e)$ of length $\ell_2$. Let $i$ be the smallest integer such that $\hat{e}_i \sim e$. We set $id(e) =  \overline{b}_i^{p, q}$. We now make the set $S$ of strings in our instance. For each endpoint $\hat{e}_i \in \hat{E}(G)$ we make a string $s_i$ as follows. 
$$s_i = a_i^1 \circ a_i^2 \ldots \circ a_i^k \circ b_i^{1,2} \circ b_i^{1,3} \ldots b_i^{1,k} \circ b_i^{2,3} \circ b_i^{2,4} \ldots \circ b_i^{k-1,k}$$
Here $a_i^p = id(u)$ if $\hat{e}_i = (u,v) \in \hat{E}_p$ and $a_i^p = \overline{a}_i^p$ otherwise. Also, $b_i^{p, q} = id(uv)$ if $\hat{e}_i \sim uv$, $u \in V_p$ and $v \in V_q$. Otherwise $b_i^{p, q} = \overline{b}_i^{p, q}$. We refer to $a_i^1$ through $a_i^k$ as the {\em vertex blocks} of $s_i$ and the $b_i^{p, q}$'s are the {\em edge blocks} of $s_i$. We refer to $a_i^p$s as the p'th vertex block and to the the $b_i^{p, q}$s as the $(p, q)$'th edge block. We set $n^* = 2{k\choose 2}$, $L = k\cdot \ell_1 + {k \choose 2}\cdot \ell_2$, and $N = |S| = 2m$, this concludes the construction. Recall that $n^*$ is the size of the solution $S^*$ sought for and observe that $L$ is the length of the constructed strings in $S$.



We consider the constructed strings $s_i$ as random variables, and for every $j$ the character $s_i[j]$ is also a random variable which takes value $1$ with probability $1/2$ and $0$ with probability $1/2$. Observe that for $j \neq j'$ and any $i$ and $i'$ the random variables $s_i[j]$ and $s_{i'}[j']$ are independent. On the other hand $s_i[j]$ and $s_{i'}[j]$ could be dependent. However, if $s_i[j]$ and $s_{i'}[j]$ are dependent then, by construction $s_i[j] = s_{i'}[j]$. Let $S^* \subset S$ such that $|S^*| = n^*$. Here we consider $S^*$ as a set of random string variables, rather than a set of strings. We are interested in studying $d(S^*, c(S^*))$ for different choices of the set $S^*$. We can write out $d(S^*, c(S^*))$ as $\sum_{p=1}^{L} d(S^*[p], c(S^*)[p])$ and so $d(S^*, c(S^*))$ is the sum of $L$ independent random variables, each taking values from $0$ to $n^*$. Thus, when $L$ is large enough $d(S^*, c(S^*))$ is sharply concentrated around $E[d(S^*, c(S^*))]$. 

We turn our attention to $E[d(S^*, c(S^*))]$ for different choices of $S^*$. The two main cases that we distinguish between is whether $S^*$ corresponds to the set of edge endpoints of a clique in $G$ or not. Before proceeding to these cases, we need some additional definitions. Let $\vec{v}$ be a vector of positive integers. We define the random variable $X_{\vec{v}} = \vec{W} \cdot \vec{v}$ where $\vec{W}$ is a random vector with same dimension as $\vec{v}$, such that each coordinate of $\vec{W}$ is drawn from $\{-1, 1\}$ uniformly at random. The variable $X_{\vec{v}}$ is interpreted as follows: start a one-dimensional random walk at $0$, in each step of the walk we go left or right with probability $1/2$. However, the length of the different steps varies, in step $i$ the walk jumps $\vec{v}[i]$ to the left or right. The value of $X_{\vec{v}}$ is the offset from the origin at the end of the walk. The total {\em length} of the random walk is $\sum_i \vec{v}[i]$ whereas the {\em number of steps} of the walk is the dimension of $\vec{v}$

Let $j$ be a position in an edge block. What we mean by this is that $s_i[j]$ is a character in $y_i^{p, q}$. Suppose no two strings of $S^*$ correspond to edge endpoints of the same edge. Then $d(S^*[j], c(S^*[j]))$ is distributed as $n^*/2 - |X_{\vec{v}}|$ where $\vec{v}$ is a $n^*$-dimensional vector of $1$s. Specifically for all $s_i \in S^*$ the $s_i[j]$s are independent so $c(S^*[j])$ is the majority character out of $n^*$ characters independently drawn from $\{0,1\}$, and $d(c(S^*[j], S^*[j]))$ is the number of occurrences of the minority character. This is distributed as $n^*/2 - |X_{\vec{v}}|$.

Again, let $j$ be a position in the $(p, q)$-edge block, but now suppose that $S^*$ contains $t$ pairs of edge endpoints that correspond to the same edge in $E_{p, q}$. $S^*$ can also contain single endpoints of edges from $E_{p, q}$ or both endpoints of edges in $E_{p',q'}$ for $(p',q') \neq (p, q)$ but we do not count these. From the construction of the $(p, q)$-edge block it follows that $d(S^*[j], c(S^*[j]))$ is distributed as $n^*/2 - |X_{\vec{v}}|$ where $\vec{v}$ is a $n^*-t$ dimensional vector with $t$ entries of value $2$ and $n^*-2t$ entries with value $1$. We define the random variable $X^i_{r,t} = i + X_{\vec{v}}$ where $v$ is a vector with $r-2t$ entries that are $1$ and $t$ entries that are $2$. Intuitively $X^i_{r,t}$ is the offset from $0$ of a random walk starting at $i$ of length $r$, with $t$ steps of length $2$ and the remaining steps of length $1$. We set $x^i_{r,t} = E[|X^i_{r,t}|]$. Finally, we define $E_{yes}$ as
\begin{align}\label{eqn:eyes}
E_{yes} = k \cdot \ell_1 \cdot (n^*/2 - x^{k-1}_{n^*-k+1,0}) + {k \choose 2} \cdot \ell_2 \cdot (n^*/2 - x^0_{n^*,1})
\end{align}

\begin{lemma}\label{lem:expyes} Let $S^*$ be a subset of $S$ of size $n^*$ that corresponds to the set of edge endpoints of a $k$-clique in $G$. Then $E[d(S^*, c(S^*))] = E_{yes}$.
\end{lemma}

\begin{proof}
For each position $j$ in a vertex block, consider the distribution of $d(S^*[j], c(S^*)[j])$. There are $k-1$ edge endpoints in $S^*$ which are all incident to the same vertex $v$, so the strings corresponding to these endpoints all have the same character at position $j$. The remaining strings all have random characters at this position. Hence $d(S^*[j], c(S^*)[j])$ is distributed as $n^*/2 - |X_{\vec{v}}|$ where $\vec{v}$ is a $n^*-(k-2)$ dimensional vector with $n^*-(k-1)$ entries of value $1$ and one entry with value $k-1$. It is easy to see that $|X_{\vec{v}}|$ is in fact distributed as $|X_{n^*-k+1,0}^{k-1}|$ since we can make the step corresponding to the entry of value $k-1$ first, and this step will take the random walk to position $k-1$ or $-(k-1)$, but with respect to distance from $0$ these positions are symmetric. Since there are $k \cdot \ell_1$ positions in vertex blocks this accounts for the first term of the equation.

For each position $j$ in an edge block $(p, q)$ there are two strings in $S^*$ that correspond to edge endpoints of the same edge in $E_{p, q}$. These two strings have the same character at position $j$. All the other strings in $S^*$ correspond to edge endpoints of strings in $E_{p',q'}$ where $p' \neq p$ or $q' \neq q$. The characters at position $j$ for these strings are drawn independently. Hence $d(S^*[j], c(S^*)[j])$ is distributed as $n^*/2 - E[|X_{n^*,1}^{0}|]$. Since there are ${k \choose 2} \cdot \ell_2$ positions in edge blocks this accounts for the second term of the equation.
\end{proof}

We now proceed to show that for any set $S^*$ that does not correspond to a set of edge endpoints of a $k$-clique in $G$, $E[d(S^*, c(S^*))]$ is at least factor $\epsilon$ greater than $E_{yes}$, where $\epsilon$ depends only on $k$. Let $\hat{E}^*$ be the set of edge endpoints corresponding to $S^*$. Define $E^*$ to be the set of edges $uv \in E(G)$ such that $(u,v) \in \hat{E}^*$ and $(v,u) \in \hat{E}^*$. Clearly, $|E^*| \leq {k \choose 2}$, hence if $E_{p, q} \cap E^* \neq \emptyset$ for every $p$, $q$ then $|E_{p, q} \cap E^*|=1$ for every $p$,$q$. We start by proving that if there exists a $p$, $q$ such that $E_{p, q} \cap E^* = \emptyset$ then $E[d(S^*, c(S^*))]$ is big. This proof is based on ``differentiating'' $x^0_{n^*,t}$ with respect to $t$. In particular for integers $i$, $r$, $t$ such that $r \geq 1$ and $t \geq 2$ define $\delta x^i_{r,t} =  x^i_{r,t} - x^i_{r,t-1}$. 
\begin{claim}\label{clm:randwalk} $x^0_{n^*,0} < x^0_{n^*,1}$. If $n^*$ is divisible by $4$ then $\delta x^0_{n^*,1} > \delta x^0_{n^*,t}$ for all $t > 1$. Furthermore, for every $i$,$t$ and $r$ we can compute $x^i_{r,t}$ in time polynomial in $i$ and $r$.
\end{claim}

The intuition of Claim~\ref{clm:randwalk} is as follows. A random walk with double steps is just the sum of independent random variables, with variables corresponding to single steps taking values from $\{-1,1\}$ and variables corresponding to double steps taking values from $\{-2,2\}$. A double step has higher variance than the sum of two single steps. Hence, if we do a random walk starting from $0$ of total length $n^*$ with $t$ double steps, then the expected distance from $0$ should increase as $t$ increases. Furthermore, as $t$ increases the variance of the random walk increses linearly, so the standard deviation increases less and less with each increment of $t$. Thus it is natural to expect that as $t$ increases, each successive step increases the expected offset from $0$ less and less. Quite surprisingly this does not hold in general (we do not prove this, as it is not important for our results). However, when the length of the random walk is a multiple of $4$, the claim does hold.

\begin{proof}[Proof of Claim~\ref{clm:randwalk}]
Recall that $x^i_{r,t}=E[|X^i_{r,t}|]$ where $X^i_{r,t}$ is a random variable denoting the final position of a random walk of length $r$, with $t$ double steps, starting at $i$. Here $i$ is an integer and might be negative. Conditional expectation yields the following recurrence for $x^i_{r,t}$, $r \geq 2t \geq 0$.
\begin{equation*}
x^i_{r,t} =
\begin{cases}
|i| & \text{if } r = 0,\\
(x^{i+1}_{r-1,t}+x^{i-1}_{r-1,t})/2 & \text{if } r > 2t,\\
(x^{i+2}_{r-2,t-1}+x^{i-2}_{r-2,t-1})/2 & \text{if } t \geq 1.
\end{cases}
\end{equation*}
It is easy to see that one of the three cases must apply when $r \geq 2t \geq 0$ - and $x^i_{r,t}$ is only defined for these values. Observe that if $r > 2t$ and $t \geq 1$ then both the second and the third case apply. The recurrence above also yields a polynomial time algorithm to compute $x^i_{r,t}$. The recurrence above together with definition of $\delta x^i_{r,t}$ yields the following recurrence for $\delta x^i_{r,t}$, for $r \geq 2t$ and $t \geq 1$.
\begin{equation*}
\delta x^i_{r,t} =
\begin{cases}
0 & \text{if } r = 2, |i| \geq 2,\\
1/2 & \text{if } r = 2,|i| = 1,\\
1 & \text{if } r = 2, |i| = 0,\\
(\delta x^{i+1}_{r-1,t}+\delta x^{i-1}_{r-1,t})/2 & \text{if } r > 2t,\\
(\delta x^{i+2}_{r-2,t-1}+\delta x^{i-2}_{r-2,t-1})/2 & \text{if } t \geq 2.
\end{cases}
\end{equation*}
A straightforward induction using this recurrence shows that $\delta x^0_{r,1} > 0$ for all $r \geq 0$, proving that $x^0_{n^*,0} < x^0_{n^*,1}$. Define $\delta^2 x^i_{r,t} = \delta x^i_{r,t} - \delta x^i_{r,t-1}$. Observe that $\delta^2 x^i_{r,t}$ is only well defined when $r \geq 2t$ and $t \geq 2$. Inserting the recurrence for $\delta x^i_{r,t}$ into the definition of $\delta^2 x^i_{r,t}$ yields the following recurrence for $\delta^2 x^i_{r,t}$.
\begin{equation}\label{eqn:delta2}
\delta^2 x^i_{r,t} =
\begin{cases}
0 & \text{if } r = 4, |i| \geq 4,\\
1/8 & \text{if } r = 4,|i| = 3,\\
1/4 & \text{if } r = 4, |i| = 2,\\
-1/8 & \text{if } r = 4,|i| = 1,\\
-1/2 & \text{if } r = 4, |i| = 0,\\
(\delta^2 x^{i+1}_{r-1,t}+\delta^2 x^{i-1}_{r-1,t})/2 & \text{if } r > 2t,\\
(\delta^2 x^{i+2}_{r-2,t-1}+\delta^2 x^{i-2}_{r-2,t-1})/2 & \text{if } t \geq 3.
\end{cases}
\end{equation}
We prove that if $r$ is divisible by $4$ then $\delta^2 x^0_{r,2} < 0$ and for all $t > 2$ we have $\delta^2 x^0_{r,t} \leq 0$. These two facts prove that for $t \geq 2$ we have
$$\delta x^0_{r,t} = \delta x^0_{r,1} + \sum_{j=2}^t \delta^2 x^0_{r,j} < \delta x^0_{r,1},$$
which is precisely the last statement of the claim. 

For integers $i$, $r \geq 0$, $t$ such that $r \geq 2t$ define $w^i_{r,t}$ to be the number of one dimensional walks of length $r$ with $t$ double steps and $r-2t$ unit steps that start in $0$ and end in $i$. Observe that $w^i_{r,t} = w^{-i}_{r,t}$. For even $r \geq 4$, expanding Equation~\ref{eqn:delta2} for $\delta^2 x^0_{r,t}$ exhaustively yields the following expression.
\begin{align*}
\delta^2 x^0_{r,t} = \Big{[}-\frac{1}{2}w^0_{r-4,t-2} + \frac{1}{4}w^2_{r-4,t-2} + \frac{1}{4}w^{-2}_{r-4,t-2}\Big{]} \Big{/} 2^{r-4} = \Big{[}-w^0_{r-4,t-2} + w^2_{r-4,t-2}\Big{]} \Big{/} 2^{r-3}.
\end{align*}
Hence to prove the statement of the claim it suffices to show that if $r$ is divisible by $4$ then $w^0_{r,0} > w^2_{r,0}$ and $w^0_{r,t} \geq w^2_{r,t}$ for $t \geq 1$. For non-negative $i$, the number of walks satisfies the following recurrence. The number of walks satisfies the following recurrence.
\begin{equation}\label{eqn:walks}
w^i_{r,t} =
\begin{cases}
1 & \text{if } r = 0,i = 0,\\
0 & \text{if } r = 0, i \neq 0,\\
w^{i-1}_{r-1,t}+w^{i+1}_{r-1,t} & \text{if } r > 2t, \\
w^{i-2}_{r-2,t-1}+w^{i+2}_{r-2,t-1} & \text{if } t > 1, i \geq 2\\
\end{cases}
\end{equation}
It is easy to see that when $r$ is even, $w^{2i}_{r,0} = {r \choose \frac{r}{2}-i}$. Since ${{r} \choose {x}} > {{r} \choose {x-1}}$ when $x \leq r/2$ it follows that
\begin{align}\label{eqn:walksNoDouble} 
w^{2i}_{r,0} > w^{2(i+1)}_{r,0} \text{for all } i
\end{align}
Equation~\ref{eqn:walksNoDouble} directly implies $w^0_{r,0} > w^2_{r,0}$.

It remains to prove that if $r$ is divisible by $4$ then $w^{0}_{r,t} \geq w^{2}_{r,t}$. For the case that $r=2t$, expanding Equation~\ref{eqn:walks} exhaustively yields the following expression.
\begin{equation}\label{eqn:walksOnlyDouble}
w^i_{2t,t} =
\begin{cases}
{t \choose (2t-i)/4} & \text{if } i \equiv 2t \text{ (mod $4$)}, \\
0 & \text{otherwise} \\
\end{cases}
\end{equation}
Most importantly, if $0 \leq i \leq i'$ and $w^i_{2t,t}$ is non-zero, then $w^i_{2t,t} \geq w^{i'}_{2t,t}$. 

We now prove that when $r-2t$ is an even, positive integer and $i \geq 0$ then $w^{2i}_{r,t} \geq w^{2(i+1)}_{r,t}$. A special case of this inequality is that when $r$ is divisible by $4$ then $w^{0}_{r,t} \geq w^{2}_{r,t}$. Observe that when $t=0$ the inequality follows by Equation~\ref{eqn:walksNoDouble}. We prove the inequality by induction on $r-t$. Observe that when $r$ decreases by $2$ while $t$ decreases by $1$, $r-t$ decreases. Hence for $t \geq 1$ and $i \geq 1$ we have 
\begin{align*}
w^{2i}_{r,t} = w^{2i-2}_{r-2,t-1} + w^{2i+2}_{r+2,t-1} \geq w^{2i}_{r-2,t-1} + w^{2i+4}_{r+2,t-1} = w^{2(i+1)}_{r,t}.
\end{align*}
Now, for $t \geq 1$ and $i = 0$ we have that $w^0_{r,t} = 2w^0_{r-2,t}+2w^2_{r-2,t}$ and $w^2_{r,t} = w^0_{r-2,t}+2w^2_{r-2,t}+ w^4_{r-2,t}$. Hence to prove that $w^0_{r,t} \geq w^2_{r,t}$ it suffices to prove $w^0_{r-2,t} \geq w^4_{r-2,t}$. If $r-2t = 2$ then by Equation~\ref{eqn:walksOnlyDouble} we have that either $w^0_{r-2,t} = w^4_{r-2,t} = 0$ or  $w^0_{r-2,t} \geq w^4_{r-2,t}$. In both cases this implies $w^0_{r,t} \geq w^2_{r,t}$. Finally, if $r-2t > 2$ then the induction hypothesis yields
$$w^0_{r,t} = 2w^0_{r-2,t}+2w^2_{r-2,t} \geq w^0_{r-2,t}+2w^2_{r-2,t}+ w^4_{r-2,t} \geq w^2_{r,t}.$$
Hence, when $r$ is divisible by $4$ then $w^{0}_{r,t} \geq w^{2}_{r,t}$, concluding the proof of the claim.
\end{proof}

Set $\Delta = \min_{i \leq n^*} (\delta x^0_{n^*,1} - \delta x^0_{n^*,i})$. By Claim~\ref{clm:randwalk}, $\Delta > 0$. Define 
\begin{align}\label{eqn:eno1}
E_{no}^1 = {k \choose 2} \cdot \ell_2 \cdot (n^*/2 - x^0_{n^*,1}) + \ell_2\Delta
\end{align}
Observe that if $\ell_2 > \frac{\ell_1 \cdot k \cdot (n^*/2)}{\Delta}$ then $E_{no}^1 > E_{yes}$. Selecting $\ell_2$ slightly larger than this will ensure the desired gap between $E_{no}^1$ and $E_{yes}$, so we set
\begin{align}\label{eqn:l2froml1}
\ell_2 = \ell_1 \cdot \left\lceil\frac{k \cdot n^*}{\Delta}\right\rceil.
\end{align}
Observe that the ratio between $\ell_2$ and $\ell_1$ is a function of $k$.

\begin{lemma}\label{lem:goodSet1} 
Let $S^*$ be a subset of $S$ of size $n^*$, where the corrresponding edge set $E^*$ has the property that $|E^*\cap E_{p, q}| \neq 1$ for at least one pair $p, q \leq k$. Then $E[d(S^*, c(S^*))] \geq E_{no}^1$.

\end{lemma}
\begin{proof}
For any position $j$ in an edge block number $p$, $q$, $d(S^*[j], c(S^*)[j])$ is distributed as $X^0_{n^*,t[p, q]}$ where $t[p, q]$ is the number of edges $e$ in $E_{p, q}$ such that for both endpoints of $e$ the strings corresponding to them are in $S^*$. It follows that $E[d(S^*, c(S^*))] \geq \ell_2 \cdot\sum_{p, q} x^0_{n^*,t[p, q]}$ since here we are just counting the contribution of the edge block positions to the expectation. Since $|S^*|=n^*$ it follows that $\sum_{p, q} x^0_{n^*,t[p, q]} \leq {k \choose 2}$. We can now use Claim~\ref{clm:randwalk} to lower bound the expectation of $d(S^*, c(S^*))$. In particular, we have that
\begin{align*}
E[d(S^*, c(S^*))] & \geq \sum_{p, q}\big{(}\ell_2 \cdot (n^*/2 - x^0_{n^*,t[p, q]}) \big{)} 
 = \sum_{p, q}\big{(}\ell_2 \cdot (n^*/2 - x^0_{n^*,1} + x^0_{n^*,1} - x^0_{n^*,t[p, q]}) \big{)} \\
& = {k \choose 2} \cdot \ell_2 \cdot (n^*/2 - x^0_{n^*,1}) + \ell_2 \cdot ({k \choose 2} x^0_{n^*,1} - \sum_{p, q}x^0_{n^*,t[p, q]}) \\
& = {k \choose 2} \cdot \ell_2 \cdot (n^*/2 - x^0_{n^*,1}) + \ell_2 \cdot \Big{(}{k \choose 2} x^0_{n^*,1} - {k \choose 2} x^0_{n^*,0} - \sum_{p, q}\sum_{t=1}^{t[p, q]} \delta x^0_{n^*,t}\Big{)} \\
& = {k \choose 2} \cdot \ell_2 \cdot (n^*/2 - x^0_{n^*,1}) + \ell_2 \cdot \Big{(}{k \choose 2} \delta x^0_{n^*,1} - \sum_{p, q}\sum_{t=1}^{t[p, q]} \delta x^0_{n^*,t}\Big{)}
\end{align*}
Observe that if $t[p, q]=1$ for all $p$, $q$ then 
$$\sum_{p, q}\sum_{t=1}^{t[p, q]} \delta x^0_{n^*,t} = {k \choose 2} \delta x^0_{n^*,1}$$
and the second term of the last equation cancels. However this is not the case, since $S^*$ is assumed {\em not} to correspond to the set of endpoints of a set of edges that intersects with every $E_{p, q}$. It follows that 
$$ \sum_{p, q}\sum_{t=1}^{t[p, q]} \delta x^0_{n^*,t} \leq {k \choose 2} \delta x^0_{n^*,1} - \Delta$$
which in turn implies that 
$$E[d(S^*, c(S^*))] \geq {k \choose 2} \cdot \ell_2 \cdot (n^*/2 - x^0_{n^*,1}) + \ell_2\Delta = E_{no}^1$$
\end{proof}

By Lemma~\ref{lem:goodSet1} we know that any set $S^*$ such that $E[d(S^*, c(S^*))] < E_{no}^1$ corresponds to all the endpoints of an edge set $E^*$ such that for every $p, q \leq k$ we have $E^* \cap E_{p, q} \neq \emptyset$. It remains to prove that if $E^*$ does not correspond to the edge set of a $k$-clique in $G$ then $E[d(S^*, c(S^*))] \geq E_{no}^2$ for an integer $E_{no}^2$ which is sufficiently large compared to $E_{yes}$. Observe that for each $i \leq k$ there are exactly $k-1$ edges in $E^*$ that are incident to vertices of $V_i$. What we prove is that if $E[d(S^*, c(S^*))] < E_{no}^2$ then for every $i$, the set of edges in $E^*$ that have an endpoint in $V_i$ all come from the same vertex. Just as for the proof of Lemma~\ref{lem:goodSet1} we need a preliminary claim about the properties of certain random walks. Let ${\cal V}$ be the set of all vectors with all positive integer entries such that the sum of the entries is exactly $n^*$ and the sum of all the entries that are not $1$ is at most $k-1$. Let ${\cal V}' = {\cal V} \setminus \vec{u}$ where $\vec{u}$ is the vector in ${\cal V}$ with one entry equal to $k-1$. Observe that for this choice of $\vec{u}$, $E[|X_{\vec{u}}|]= x^{k-1}_{n^*-k+1,0}$. Set $\Delta_2 = \min_{\vec{v} \in {\cal V}'} x^{k-1}_{n^*-k+1,0} - E[|X_{\vec{v}}|]$

\begin{claim}\label{clm:randomwalk2} For every $\vec{v} \in {\cal V}'$, $x^{k-1}_{n^*-k+1,0} - E[|X_{\vec{v}}|] > 0$. Hence $\Delta_2$ is positive. Furthermore, $\Delta_2$ can be computed in time $f(k)$ for some function $f$. \end{claim}
\begin{proof}
To see that $\Delta_2$ can be computed in time $f(k)$ for some function $f$ it is sufficient to observe that the size of the sample space of any variable $X_{\vec{u}}$ is bounded by a function of $k$ and that $|\cal V|$ is upper bounded by a function of $k$ as well. We now prove that $\Delta_2$ is positive. Consider a vector $\vec{v} \in {\cal V}$ and let $\vec{v}'$ be a vector that contains all the entries of $\vec{v}$ that are greater than $1$, and possibly some entries that are $1$ such that the sum of the entries in $\vec{v}'$ is exactly $k-1$. Conditional expectation yields:
$$E[|X_{\vec{v}}|] = \sum_{i=-k+1}^{k-1} P[X_{\vec{v}'} = i] \cdot x^i_{n^*-k+1,0}.$$
Observe that for every $i$ whose parity is not the same as $k-1$ we have that $P[X_{\vec{v}'} = i] = 0$ since every entry of $\vec{v}'$ is either added or subtracted to get $X_{\vec{v}'}$ and $-1 \equiv 1 (\mod 2)$. Furthermore, a simple induction (see below) shows that for every non-negative $i'$ such that $|i'| \leq |i|-2$, $x^i_{r,0} > x^{i'}_{r,0}$. Together these two facts imply that for all $i \notin \{k-1,-k+1\}$ for which $P[X_{\vec{v}'} = i]$ is non-zero we have $x^i_{n^*-k+1,0} < x^{k-1}_{n^*-k+1,0}$. Since such an $i$ must exist for any $\vec{v} \in {\cal V}'$ we have that $x^{k-1}_{n^*-k+1,0} - E[|X_{\vec{v}}|] > 0$.

Finally, we have to prove that for every non-negative $i' \leq i-2$, $x^i_{r,0} > x^{i'}_{r,0}$. We do this by induction on $r$. For $r=0$ this clearly holds as $x^i_{0,0}=|i|$. For $r \geq 1$ conditional expectation yields that $x^i_{r,0} = 1/2(x^{i-1}_{r-1,0}+x^{i+1}_{r-1,0}) > 1/2(x^{i'-1}_{r-1,0}+x^{i'+1}_{r-1,0}) = x^{i'}_{r,0}$ if $i' \neq 0$ and $x^i_{i,0} = 1/2(x^{i-1}_{r-1,0}+x^{i+1}_{r-1,0}) > 1/2(2x^{1}_{r-1,0}) = x^{i'}_{r,0}$ if $i'=0$. This concludes the proof.
\end{proof}

We are now ready to prove the last part of the reduction. Let
\begin{align}\label{eqn:eno2}
E_{no}^2 =  {k \choose 2} \cdot \ell_2 \cdot (n^*/2 - x^0_{n^*,1}) + k \cdot \ell_1 \cdot (n^*/2 - x^{k-1}_{n^*-k+1,0}) + \ell_1\Delta_2.
\end{align}

\begin{lemma}\label{lem:noCliqueExp} Let $S^*$ be a subset of $S$ of size $n^*$ that corresponds to all edge endpoints of a set $E^*$ of edges such that for every $p, q \leq k$ we have $|E^* \cap E_{p, q}|=1$. If there exist two distinct vertices $v_1, v_2 \in V_i$ such that $E^*$ contains edges incident to both $v_1$ and $v_2$ then $E[d(S^*, c(S^*))] \geq E_{no}^2$.
\end{lemma}
\begin{proof}

Consider a set $S^*$ satisfying the conditions of the lemma. The contribution to $E[d(S^*, c(S^*))]$ of the edge blocks is exactly ${k \choose 2} \cdot \ell_2 \cdot (n^*/2 - x^0_{n^*,1})$. Now, consider a vertex block $i$ such that there is exactly one vertex $v_i \in V_i$ that is incident to edges of $E^*$. This block contributes exactly $\ell_1 \cdot (n^*/2 - x^{k-1}_{n^*-k+1,0})$ to $E[d(S^*, c(S^*))]$. Finally, consider a block $i$ such that there are two distinct vertices $v_1, v_2 \in V_i$ such that $E^*$ contains edges incident to both $v_1$ and $v_2$. Let $\vec{v}$ be a vector that for each vertex $v \in V_i$ which is incident to at least one edge in $E^*$, contains an entry which is exactly the number of edges in $E^*$ that $v$ is incident to. For each edge which is not incident to any vertices in $V_i$ the vector $\vec{v}$ contains a entry with value $1$. Hence the sum of the entries of $\vec{v}$ is $n^*$, the sum of all entries in $\vec{v}$ that are greater than $1$ is at most $k-1$, and $\vec{v}$ contains no entry which is $k-1$. This is because exactly $k-1$ edges are incident to vertices in $V_i$ and two such edges are incident to distinct vertices. Hence $\vec{v} \in {\cal V}'$. Now, observe that the vertex block $i$ contributes exactly $\ell_1 \cdot (n^*/2 - E[|X_{\vec{v}}|])$ to $E[d(S^*, c(S^*))]$. By Claim~\ref{clm:randomwalk2}, $E[|X_{\vec{v}}|] \leq x^{k-1}_{n^*-k+1,0} - \Delta_2$. Thus,
$$E[d(S^*, c(S^*))] \geq {k \choose 2} \cdot \ell_2 \cdot (n^*/2 - x^0_{n^*,1}) +  k \cdot \ell_1 \cdot (n^*/2 - x^{k-1}_{n^*-k+1,0}) + \ell_1\Delta_2 = E_{no}^2$$
\end{proof}

Set $E_{no} = \min(E_{no}^1, E_{no}^2) = E_{no}^2$. From Equations~\ref{eqn:eyes}, \ref{eqn:eno1}, \ref{eqn:eno2} and \ref{eqn:l2froml1} we conclude that there exist constants $\kappa_{yes}$, $\kappa_{no}$ and $\kappa_{L}$ depending only on $k$ such that $E_{yes} = \kappa_{yes}\ell_1$, $E_{no} = \kappa_{no}\ell_1$ and $L = \kappa_L\ell_1$. Furthermore, $\kappa_{yes} < \kappa_{no}$ and the value of $\kappa_{yes}$, $\kappa_{no}$ and $\kappa_{L}$ can be computed in time $f(k)$ for some function $f$. Set $\kappa_{yes}' = (2\kappa_{yes}+\kappa_{no}) / 3$ and $\kappa_{no}' = (\kappa_{yes}+2\kappa_{no}) / 3$. Then $\kappa_{yes} < \kappa_{yes}' < \kappa_{no}' < \kappa_{no}$. We set $D_{yes} = \kappa_{yes}'\ell_1$ and $D_{no} = \kappa_{no}'\ell_1$. 

\paragraph{A randomized analogue of Theorem~\ref{thm_no_eptas}.}
Before proving Theorem~\ref{thm_no_eptas} we argue that the randomized construction works. Specifically, we show that if {\em Min-distance Consensus String With Outliers} has an EPTAS then {\em MCC} has a randomized FPT algorithm, implying that W[1] $\subseteq$ randomized FPT. The results proved in this section are not used in the proof of Theorem~\ref{thm_no_eptas}, but they provide useful insights on how the deterministic construction works.
\begin{lemma}\label{lem:noeptas:concentration} For any $S^* \subset S$ such that $|S^*|=n^*$, 
\begin{align*}
P\Big{[}|d(S^*, c(S^*))-E[d(S^*, c(S^*))]| > x \cdot \ell_1\Big{]} \leq 2\exp\left(-2\frac{x^2}{\kappa_L (n^*)^2} \ell_1\right).
\end{align*}
\end{lemma}
\begin{proof}
We have that $d(S^*, c(S^*)) = \sum_{p=1}^{L} d(S^*[p], c(S^*)[p])$. The $d(S^*[p], c(S^*)[p])$'s are independent random variables taking values from $0$ to $n^*$. Since $L=\kappa_L\ell_1$ it follows that $P[|d(S^*, c(S^*))-E[d(S^*, c(S^*))]| > x \cdot \ell_1] = P[|d(S^*, c(S^*))-E[d(S^*, c(S^*))]| > \frac{x}{\kappa_L} \cdot L]$. By Hoeffding's inequality (Proposition~\ref{prop:hoeff}) it follows that
\begin{align*}
P\left[|d(S^*, c(S^*))-E[d(S^*, c(S^*))]| > \frac{x}{\kappa_L} \cdot L\right] & \leq 2 \exp\left(-2\left(\frac{x}{\kappa_Ln^{*}}\right)^2L\right) \\
 & = 2\exp\left(-2\frac{x^2}{\kappa_L (n^*)^2} \ell_1\right)
\end{align*}
\end{proof}
We now define $\ell_1$. This value for $\ell_1$ is only valid for the randomized construction, and a different value for $\ell_1$ is used in the proof of Theorem~\ref{thm_no_eptas}.
\begin{align}\label{eqn:ell1random}
\ell_1 = \frac{(n^*)^2\kappa_L}{2(\kappa_{yes}' - \kappa_{yes})}\ln\left(20(2m)^{n^*}\right).
\end{align}
Recall that $m$ is the number of edges in the graph $G$, so $m \leq n^2$ and hence $\ell_1 \leq f \cdot \log n$ for some $f$ depending only on $k$.
\begin{lemma}\label{lem:randConstr} If $G$ has a $k$-clique $C$, let $S^*$ be the set of strings corresponding to edge endpoints of edges in $C$. Then $P[d(S^*, c(S^*)) > D_{yes}] \leq \frac{1}{10(2m)^{n^*}}$. If $G$ does not contain a $k$-clique, then the probability that $S$ contains a subset $S^*$ of size $n^*$ such that $d(S^*, c(S^*)) < D_{no}$ is at most $1/10$. \end{lemma}
\begin{proof}
If $G$ has a $k$-clique $C$, let $S^*$ be the set of strings corresponding to edge endpoints of edges in $C$. Then by Lemma~\ref{lem:expyes}, $E[d(S^*, c(S^*))] = D_{yes}$. Now, $D_{yes}-E_{yes} = (\kappa_{yes}' - \kappa_{yes})\ell_1$ and hence, by Lemma~\ref{lem:noeptas:concentration}, 
$$P[d(S^*, c(S^*)) > D_{yes}] \leq P[|d(S^*, c(S^*))-E_{yes}| > (\kappa_{yes}' - \kappa_{yes})\ell_1] \leq \frac{1}{10(2m)^{n^*}}.$$
On the other hand, consider a set $S^*$ of size $n^*$ that does not correspond to the edge endpoints of a clique. If $S^*$ does not correspond to a set $E^*$ of edges such that $|E^* \cap E_{p, q}|=1$ for every $p$,$q$, then $E[d(S^*, c(S^*))] \geq E_{no}^1 > E_{no}$. If $S^*$ corresponds to a set $E^*$ of edges such that $|E^* \cap E_{p, q}|=1$, but $E^*$ is not the edge set of a clique in $G$ then there exists an $i$ and $v_1$, $v_2 \in V_i$ such that $E ^*$ contains edges incident to $v_1$ and to $v_2$. In this case Lemma~\ref{lem:noCliqueExp} yields that $E[d(S^*, c(S^*))] \geq E_{no}^2 = E_{no}$. Hence $E[d(S^*, c(S^*))] \geq E_{no}$. Finally, $E_{no}-D_{no} = (\kappa_{no}-\kappa_{no}')\ell_1$ and $(\kappa_{no}-\kappa_{no}')\ell_1 = (\kappa_{yes}'-\kappa_{yes})\ell_1$ and hence, by Lemma~\ref{lem:noeptas:concentration}, 
$$P[d(S^*, c(S^*)) \leq D_{no}] = P[E_{no}-d(S^*, c(S^*)) > (\kappa_{yes}' - \kappa_{yes})\ell_1] \leq \frac{1}{10(2m)^{n^*}}.$$
Thus, if $G$ does not contain a clique of size $k$, the union bound yields that the probability that $S$ contains a subset $S^*$ of size $n^*$ such that $d(S^*, c(S^*)) < D_{no}$ is at most $1/10$.
\end{proof}
We now prove a randomized analogue of Theorem~\ref{thm_no_eptas}.
\begin{lemma}\label{lem:randNoEptas}
If {\em Min-distance Consensus String With Outliers} has an EPTAS then W[1] $\subseteq$ randomized FPT.
\end{lemma}
\begin{proof}
Assiming that {\em Min-distance Consensus String With Outliers} has an EPTAS we give a randomized fixed parameter tractable algorithm for {\em MCC} with two sided error. We construct the instance to {\em Min-distance Consensus String With Outliers} as described and run the EPTAS with $\epsilon = \frac{D_{no}}{D_{yes}} - 1 = \frac{\kappa_{no}'}{\kappa_{yes}'} - 1$. If the EPTAS returns a set $S^*$ such that $d(S^*, c(S^*)) \leq D_{no}$ the algorithm returns that the input graph $G$ contains a $k$-clique, otherwise we return that $G$ has no $k$-clique. The construction takes time $O(f(k)n^{O(1)})$ for some function $f$, and $\epsilon$ depends only on $k$. Hence the EPTAS runs in time $g(k)n^c$ for some function $g$. Thus the algorithm terminates in FPT time.

If $G$ contains a $k$-clique, then by Lemma~\ref{lem:randConstr}, with probability at least $1-\frac{1}{10(2m)^{n^*}} \geq 1-\frac{1}{n^k}$ there is a set $S^*$ of size $n^*$ such that $d(S^*, c(S^*)) \leq D_{yes}$. If this event occurs, the EPTAS will find a solution $S'$ such that $d(S', c(S')) \leq D_{yes}(1+\epsilon) \leq D_{no}$ and hence the algorithm will correctly return ``yes''. Hence the probability of false negatives is at most $\frac{1}{n^k}$.

If $G$ does not contain a $k$-clique, then by Lemma~\ref{lem:randConstr}, with probability at least $9/10$ for every set $S^*$ of size $n^*$ we hae $d(S^*, c(S^*)) > D_{no}$. If this event occurs the algorithm correctly returns ``no'' and hence the probability if false positives is at most $1/10$. This implies that there is a randomized fixed parameter tractable algorithm for {\em MCC}, which in turn shows that  W[1] $\subseteq$ randomized FPT.
\end{proof}

\paragraph{A Deterministic Construction and Proof of Theorem~\ref{thm_no_eptas}.}
In order to prove Theorem~\ref{thm_no_eptas} we need to make the construction deterministic. We only used randomness to construct the set $Z$, all other steps are deterministic. We now show how $Z$ can be computed deterministically instead of being selected at random, preserving the properties of the reduction. For this, we need the concept of near $p$-wise independence defined by Naor and Naor~\cite{NaorN93}. The original definition of near $p$-wise independence is in terms of sample spaces, we define near $p$-wise independence in terms of collections of binary strings. This is only a notational difference, and one may freely translate between the two variants.
\begin{definition}[\cite{NaorN93}]
A set $C = \{c_1, c_2, \ldots c_t\}$ of length $\ell$ binary strings is $(\epsilon,p)$-independent if for any subset $C'$ of $C$ of size $p$, if a position $i \leq t$ is selected uniformly at random, then
$$\sum_{\alpha \in \{0,1\}^p} |P[C'[i] = \alpha] - 2^{-p}| \leq \epsilon.$$
\end{definition}
Naor and Naor~\cite{NaorN93} and Alon et al.~\cite{AlonGHP92} give determinsitic constructions of small nearly $k$-wise independent sample spaces. Reformulated in our terminology, Alon et al. prove a slightly stronger version of the following theorem.
\begin{theorem}[\cite{AlonGHP92}]\label{thm:nearlyindep} For every $t$, $p$, and $\epsilon$ there is a $(\epsilon,p)$-independent set $C = \{c_1, c_2, \ldots c_t\}$ of binary strings of length $\ell$, where $\ell = O(\frac{2^k \cdot k\log t}{\epsilon})$. Furthermore, $C$ can be computed in time $O(|C|^{O(1)})$.
\end{theorem}
We use Theorem~\ref{thm:nearlyindep} to construct the set $Z$. We set
$$\epsilon = \frac{\kappa_{yes}'-\kappa_{yes}}{\kappa_L \cdot n^*}$$
and construct an $(\epsilon,n^*)$-independent set $C$ of $2m$ strings. These strings have length $\ell = f \cdot \log(n)$ for some $f$ depending only on $k$, and $C$ can be constructed in time $O(gn^{O(1)})$ for some $g$ depending only on $k$. We set $\ell_1 = \ell$. Observe that since $\ell_2$ is an integer multiple of $\ell_1$, the length of the strings in $Z$ is an integer multiple of $\ell_1$. For every $i$ we set $z_i = c_i \circ c_i \circ \ldots \circ c_i$, where we used $\kappa_L$ copies of $c_i$ such that $z_i$ is a string of length $L$. The remaining part of the construction, i.e the construction of $S$ from $Z$ remains unchanged. To distinguish between the deterministically constructed $S$ and the randomized construction, we refer to the deterministically constructed $S$ as $S_{det}$. We now prove that for every $S^*_{det} \subseteq S_{det}$ of size $n^*$, if $S^*$ is the set of strings in the randomized construction that corresponds to the same edge endpoints as $S^*_{det}$, then $d(S^*_{det}, c(S^*_{det}))$ is almost equal to $E[d(S^*, c(S^*))]$.

For a subset $I$ of $\{1,2,\ldots,2m\}$ define $S^*(I) = \{s_i \in S~:~i \in I\}$ and $S^*_{det}(I) = \{s_i \in S_{det}~:~i \in I\}$. For every $j \leq \kappa_L$, define $P_j = \{\kappa_L \cdot j + 1,\kappa_L \cdot (j+1)\}$. Hence for every $i$ and $j$, $z_i[P_j] = c_i$. The construction of $S_{det}$ (and $S$) from $Z$ implies that for every $j$, the substring there exists a function $f_j : \mathbb{N} \rightarrow \mathbb{N}$ such that for any $i \leq 2m$, $s_i[P_j] = z_{f(i)}[P_j]$. For any $I \subseteq \{1,2,\ldots,2m\}$ and $j < \kappa_L$ we define $Z^*(I,j) = \{z_{f_j(i)}~:~i \in I\}$. This means that for a subset $I \subseteq \{1,2,\ldots,2m\}$ of size $n^*$, the set $Z^*(I,j)$ is the set of $n^*$ strings in $Z$ which $S^*(I)[P_j]$ and $S^*_{det}(I)[P_j]$ depend on. For every set $I \subseteq \{1,2,\ldots,2m\}$ of size $n^*$ and integer $j < \kappa_L$ define $d_j^{I} : \{0,1\}^{n^*} \rightarrow \{0,1,\ldots,n^*\}$ to be a function such that for any $p \in P_j$, if $Z^*(I) = \alpha$ then $d(S^*(I)[p],c(S^*(I)[p])) = d_j^I(\alpha)$ and $d(S_{det}^*(I)[p],c(S_{det}^*(I)[p])) = d_j^I(\alpha)$. Since $S^*(I)[p]$ depends in exactly the same way on $Z^*(I)[p]$ for all $p \in P_j$ the function $d_j^I$ is well defined. For every set $I \subseteq \{1,2,\ldots,2m\}$ of size $n^*$ and integer $j < \kappa_L$ we have the following expression for $d(S^*(I)_{det}[P_j],c(S^*(I)_{det}[P_j]))$.
\begin{align}\label{eqn:derand1}
d\left(S^*_{det}(I)[P_j],c(S^*_{det}(I)[P_j])\right) = \ell_1 \cdot \sum_{\alpha \in \{0,1\}^{n^*}} P[Z^*(I)[p] = \alpha] \cdot d_j^I(\alpha)
\end{align}
Here the probability $P[Z^*(I)[p] = \alpha]$ is taken when $p$ is selected from $P_j$ uniformly at random. For the randomized construction we have that $P[Z^*(I)[p] = \alpha] =  \frac{1}{2^{n^*}}$, which yields the following expression.
\begin{align}\label{eqn:derand2}
E\left[d\left(S^*(I)[P_j],c(S^*(I)[P_j])\right)\right] = \ell_1 \cdot \sum_{\alpha \in \{0,1\}^{n^*}} \frac{1}{2^{n^*}} \cdot d_j^I(\alpha)
\end{align}
Combining Equations~\ref{eqn:derand1} and~\ref{eqn:derand2} yields the following bound.
\begin{align}\label{eqn:derand3}
\nonumber \Big{|}d\left(S^*_{det}(I)[P_j],c(S^*_{det}(I)[P_j])\right) & - E[d(S^*(I)[P_j], c(S^*(I)[P_j]))]\Big{|} \\
& = \ell_1 \cdot \left|\sum_{\alpha \in \{0,1\}^{n^*}} \left(P[Z^*(I)[p] = \alpha] - \frac{1}{2^{n^*}}\right)\cdot d_j^I(\alpha)\right|\\
\nonumber & \leq \ell_1 \cdot \sum_{\alpha \in \{0,1\}^{n^*}} \left|\left(P[Z^*(I)[p] = \alpha] - \frac{1}{2^{n^*}}\right)\right|\cdot n^* \\
\nonumber & \leq \ell_1 \cdot \epsilon \cdot n^*
\end{align}
Summing Equation~\ref{eqn:derand3} over $0 \leq j < \kappa_L$ yields the desired bound for every $I \subseteq \{1,2,\ldots,2m\}$ of size $n^*$.
\begin{align}\label{eqn:derand4}
\Big{|}d\left(S^*_{det}(I),c(S^*_{det}(I))\right) - E[d(S^*(I), c(S^*(I)))]\Big{|} \leq \ell_1 \cdot \kappa_L \cdot \epsilon \cdot n^* \leq \ell_1 \cdot (\kappa_{yes'}-\kappa_{yes})
\end{align}
Equation~\ref{eqn:derand4} allows us to finish the proof of Theorem~\ref{thm_no_eptas}. For any set $S^*$ of size $n^*$ that corresponds to a clique in $G$, we have that $E[d(S^*(I), c(S^*(I)))] = E_{yes} = \ell_1\kappa_{yes}$, and so by Equation~\ref{eqn:derand4}, $d\left(S^*_{det}(I),c(S^*_{det}(I))\right) \leq \ell_1\kappa_{yes}'= D_{yes}$. For any set $S^*$ of size $n^*$ that does not correspond to a clique in $G$, we have that $E[d(S^*(I), c(S^*(I)))] \geq E_{no} = \ell_1\kappa_{no}$, and so by Equation~\ref{eqn:derand4}, $d\left(S^*_{det}(I),c(S^*_{det}(I))\right) \geq \ell_1\kappa_{no}'= D_{no}$. Since $\frac{D_{no}}{D_{yes}} \geq 1+\delta$ for some $\delta$ depending only on $k$, an EPTAS for {\em Min-distance Consensus String With Outliers} can be used to distinguish between images of ``yes'' instances of {\em MCC} and images of ``no'' instances of {\em MCC} in time $f(k)n^{O(1)}$ for some function $f$. Hence {\em Min-distance Consensus String With Outliers} does not have an EPTAS unless FPT=W[1], concluding the proof of Theorem~\ref{thm_no_eptas}.
\hfill $\Box$
\section{Parameterized Intractability Results} \label{sec:intract}
From Theorem~\ref{thm_no_eptas} we can extract intractability results for various parameterizations of {\em Consensus String with Outliers}. In the proof of Theorem~\ref{thm_no_eptas} we reduced instances of {\em MCC} to an instance of {\em Consensus String with Outliers} where the size $n^*$ of the solution sought for is $k \cdot (k-1)$. Here $k$ is the size of the clique sought for in the {\em MCC} instance. Thus a FPT algorithm for {\em Consensus String with Outliers} parameterized by $n^*$ would give an FPT algorithm for {\em MCC}. This proves the following theorem.

\begin{theorem}\label{cor:n_star_w1} Consensus String with Outliers is W[1]-hard when parameterized by $n^*$, even when $\Sigma = \{0, 1\}$. \end{theorem}
Since an EPTAS for a problem implies an FPT algorithm for the problem parameterized by the value objective function~\cite{marx_survey}, Theorem~\ref{cor:n_star_w1} immediately implies that {\em Consensus String with Max Non-Outliers} does not admit an EPTAS unless FPT=W[1]. This means that in some sense the PTAS provided in Theorem~\ref{thm:max_n_star} is the best we can hope for.



In the context of error correction for DNA fragment assembly, we expect the number $k$ of outliers to be reasonably small. A simple brute force algorithm for {\em Consensus String with Outliers} that tries all ${n \choose k}$ subsets of $S$ of size $n^*$ works in $O(n^{k+O(1)}\ell)$ time. It is interesting whether we can significantly improve over this algorithm, in particular whether {\em Consensus String with Outliers} is FPT when parameterized by $k$. Using Theorem~\ref{cor:n_star_w1} as a starting point, we show that {\em Consensus String with Outliers} parameterized by $k$ is W[1]-hard, even when the alphabet is binary. 

It will be convenient to consider a set of strings $S = \{ s_1, \ldots, s_n \}$ of $n$ length-$\ell$ strings as a $n \times \ell$ matrix, then, the $i$th {\em column} of $S$ is the vector $[s_1[i], \ldots, s_n[i]]^T$.  An instance of {\em Consensus String with Outliers} is given by a set $S$ of $n$ length-$\ell$ strings with input parameters $k$ and $d$.  We assume $\ell > 2$ and $d > 2$ since $\ell \leq 1$ and $d \leq 1$ produce trivial cases.  We describe how to generate a {\em Consensus String with Outliers} instance with a set $S'$ of $n'$ strings of length $\ell'$ and parameters $d'$ and $k'$, where there exists a subset $S_o'$ of $n^*$ outlier strings such that $d(S' /S_o', c(S' / S_o')) \leq d'$ if and only if there exists a subset $S^*$ of $n^*$ non-outlier strings to the original instance such that $d(S^*, c(S^*)) \leq d$.  Let $n' = 3n$, $k' = n^*$, $d' = \frac{3n \ell}{2} - n^* \ell + d + 10 \ell  (n - n^*)$ and $\ell' = 11 \ell$.  The $3n$ strings are generated as follows: 

\begin{enumerate}
\item For each string $s_i \in S$ there exists a $11\ell$-length string $s_i' \in S'$, where the first $\ell$ symbols of $s_i'$ are equal to $s_i$ and the remaining $10\ell$ symbols of $s_i'$ are equal to 1.  We denote these subset of strings of $S'$ as $S_{org}'$ 
\item The remaining $2n$ length-$11\ell$ strings are constructed so that each of the first $\ell$ columns of $S'$ contain an equal number of positions equal to 0 and positions equal to 1. The last $10\ell$ positions are equal to 0.    
\end{enumerate}


\begin{lemma} \label{direction_one_lemma} For a Consensus String with Outliers instance containing a subset $S^*$ of $n^*$ non-outlier strings that satisfy $d(S^*, c(S^*)) \leq d$, the previous construction produces an instance with a set $S_o'$ of $k'$ outlier strings that satisfy $d(S' / S_o', c(S' / S_o')) \leq d'$. \end{lemma}

\begin{proof} 
Without loss of generality, assume $c(S^*)$ is equal to $0^{\ell}$. Let $S_o'$ be the $n^*$ strings of $S_{org}'$ corresponding to $S^*$. Then we claim $d(S' / S_o', 0^{11 \ell}) \leq d'$.  Since there exists $2n$ strings equal to 0, and $n$ strings equal to 1 at the last $10 \ell$ positions, it follows that $c(S' / S_o')[i] = 0$ for all $\ell < i \leq 11\ell$.  By our assumption that $S_o'$ is equal to $S^*$ it follows that the contribution of these $10 \ell$ positions to $d(S' / S_o', 0^{11 \ell})$ is $10 \ell (n - n^*) $.   Now consider the first $\ell$ positions, which we remind the reader that they contain an equal number of 0's and 1's.  Since we eliminate $n^*$ strings from $S'$ the contribution to $d(S' / S_o', 0^{11 \ell})$ is at most $\ell \left( \frac{3n}{2} - n^*\right)+d$, concluding the proof. 
\end{proof} 

For the reverse direction, we need to prove that the existence a subset $S_o'$ of $n^*$ outlier strings in $S'$ that satisfy the constraint $d(S' / S_o', c( S' / S_o')) \leq d'$, implies the existence of subset $S^*$ of $n^*$ strings in $S$ that satisfy the constraint $d(S^*, c( S^*)) \leq d$.  

\begin{lemma} \label{direction_two_lemma} The $k'$ outlier strings in $S'$ correspond to a subset $S^*$ of $n^*$ non-outlier strings in $S$ where $d(S^*, c( S^*)) \leq d$.   \end{lemma}

\begin{proof} 
Let $S_o'$ be a set of $n^*$ outlier strings in $S'$ that correspond to the minimum distance, {\em i.e.} $d(S' / S_o', c(S' / S_o')) \leq d(S' / S_o'', c(S' / S_o'')$ for any subset $S_o''$ that is not equal to $S_o'$.  Since there are $2n$ strings of $S'$ that are equal to 0 in the last $10 \ell$ positions, and $n$ strings of $S'$ that are nonzero at these positions, it follows that $c(S' / S_o')[i] = 0$ for all $\ell < i \leq 11\ell$.  We argue that $S_o'$ is contained in $S_{org}'$.  From the pigeonhole principle it follows that there exists at least one string, say $s_1'$, that is not in $S_o'$ but is contained $S_{org}'$. For contradiction, we assume there exists a string, say $s_2'$, that is contained in $S_o'$ but not contained in $S_{org}'$.  Note that $d(s_2, c(S' / S_o')) < d(s_1, c(S' / S_o'))$ because at the last $10\ell$ positions we have: $c(S' / S_o')$ equal to 0, $s_1$ equal to 1, and $s_2$ equal to 0.   Let $\bar{S} = \{S' / S_o'\} / s_1 \cap s_2$.  By definition of $c(\bar{S})$, we have $d(\bar{S}, c(\bar{S})) \leq  d(\bar{S}, c(S' / S_o'))$, which can be bounded as follows: 
		
\begin{align*}
d(\bar{S}, c(\bar{S})) 	& \leq d(\bar{S}, c(S' / S_o')) \\
								& = d(S' / S_o', c(S' / S_o')) - d(s_1, c(S' / S_o')) + d(s_2, c(S' / S_o')) \\
								& < d(S' / S_o', c(S' / S_o')).
\end{align*}

We contradict the fact that $S' / S_o'$ is a minimal solution solution and all outlier strings in $S_o'$ are contained in $c(\bar{S})$. The last $10\ell$ positions will have at least $10\ell(n - n^*)$ mismatches with $c(S' / S_o')$ and it follows that the bound $\sum_{s_i' \in S' / S_o'} d(0^{11\ell}, s_i') \leq d'$ is achieved when $d(S_o', c(S_o')) \leq d$.   
\end{proof} 

Our main theorem follows direction from Lemma \ref{direction_one_lemma} and Lemma \ref{direction_two_lemma}.  

\begin{theorem}\label{cor:k_w1} Consensus String with Outliers is W[1]-hard when parameterized by $k$, even when $\Sigma = \{0, 1\}$. \end{theorem}

\section{Parameterized Tractability Results} \label{sec:tract_results}
 
In this section, we prove {\em Consensus String with Outliers} is fixed-parameter tractable with respect to the parameter $\delta = d/n^*$ when the alphabet size is bounded by a constant.  For the remainder of this section, we make the assumption that $\delta >0$; otherwise  {\em Consensus String with Outliers} can trivially be solved in polynomial time. The algorithm and analysis are nearly identical to that demonstrating {\em Consensus Patterns} is fixed-parameter tractable with respect to the parameterization $\delta = d/n$ and bounded alphabet size \cite{M08}, where $\delta$ is the average error between the consensus string and the length-$\ell$ substrings $s_i'$. 
 
First, we define some terms and notation that will be used in this section. A {\em hypergraph} $G=(V_G, E_G)$ consists of a set of vertices $V_G$ and a collection of edges $E_G$, where each edge is a subset of $V_G$. Given two hypergraphs, $H=(V_H, E_H)$ and $G=(V_G, E_G)$, we say $H$ appears at $V' \subseteq V_G$ as a {\em partial hypergraph} if there is a bijection $\pi$ between the elements of $V_H$ and $V'$ such that for every edge $E \in E_H$ there exists and edge $\pi(E) \in E_G$ (where the mapping $\pi$ is extended to the edges the obvious way). For example, if $H$ has the edges $\{1, 2\}$, $\{2, 3\}$, and $G$ has the edges $\{a, b\}$, $\{b, c\}$, $\{c, d\}$, then $H$ appears as a partial hypergraph at $\{a, b, c\}$ and at $\{b, c, d\}$. Given two hypergraphs, $H=(V_H, E_H)$ and $G=(V_G, E_G)$, we say that $H$ appears at $V' \subseteq V_G$ as {\em subhypergraph} if there is such a bijection $\pi$ where for every edge $e \in E_H$, there is an edge $e' \in E_G$ with $\pi(e) = e' \cap V'$. For example, let the edges of $H$ be $\{1, 2\}$, $\{2, 3\}$, and let the edges of $G$ be $\{a, c, d\}$, $\{b, c, d\}$. 

An {\em edge cover} of $H$ is a subset $E' \subseteq E_H$ such that each vertex of $V_H$ is contained in at least one edge of $E'$. The edge cover number $\rho(H)$ is the size of the smallest edge cover in $H$. A {\em fractional edge cover} is an assignment $\Psi: E_H \rightarrow [0, 1]$ such that $\sum_{E: v \in E} \Psi(E) \geq 1$ for every vertex $v$. The {\em fractional cover number}, denoted as $\rho^*(H)$, is the minimum of $\sum_{E \in E_H} \Psi(E)$ taken over all fractional edge covers $\Psi$. 

Marx \cite{M08} demonstrated {\em Consensus Patterns} can be solved in $f(\delta) \cdot n^9$ by constructing a hypergraph $G$ from the {\em Consensus Patterns} instance, defining a combinatorial characterization of a solution to the instance with respect to the hypergraph respresentation, and enumerating (efficiently) over all subhypergraphs in $G$ with the defined combinatorial characterization.  It is shown that hypergraphs having at most $\delta$ vertices and at most $200 \log \delta$ edges need to be considered (Proposition 6.3 in \cite{M08}), and that any edge of size greater than $20 \delta$ can be removed from $G$ and all subhypergraph corresponding to a solution to the original {\em Consensus Patterns} instance can be retained, if they exist.  The enumeration step is completed by considering all possible hypergraph with at most $\delta$ vertices and at most $200 \log \delta$ edges, and for each such hypergraph, $H_0$, determining every place where $H_0$ appears in $G$ as a subhypergraph.  This paradigm for solving the {\em Consensus Patterns} problem makes use of an efficient algorithm for finding all the places $V' \subseteq V_G$ in $G$ where $H$ appears as hypergraph for two given hypergraphs $H=(V_H, E_H)$ and $G=(V_G, E_G)$. The result of Marx \cite{M08}, which proves a tight upper bound on the time required to perform this enumeration step, is essential. 

The following result by Friedgut and Kahn \cite{FK} gives a bound on the maximum number of times a hypergraph $H=(V_H, E_H)$ can appear as partial hypergraph in a hypergraph $G$ with $m$ edges, i.e.~the maximum number of different subsets $V' \subseteq V_G$ where $H$ can appear in $G$.

\begin{theorem} \label{thm:fk} \cite{FK} Let $H$ be a hypergraph with fractional cover number $\rho^*(H)$, and let $G$ be a hypergraph with $m$ edges. There are at most $|V_H|^{|V_H|} \cdot m^{\rho^*(H)}$ different subsets $V' \subseteq V_G$ such that $H$ appears in $G$ at $V'$ as partial hypergraph. Furthermore, for every $H$ and sufficiently large $m$, there is a hypergraph with $m$ edges where $H$ appears $m \cdot \rho^*(H)$ times. \end{theorem}

Marx \cite{M08} extended this theorem by giving a bound on the running time required to enumerate through all possible partial hypergraphs of a given hypergraph $G$. In particular, if $H$ is a hypergraph with fractional cover number $\rho^*(H)$, and $G$ is a hypergraph with $m$ edges and the size of each edge is at most $\ell$ then hypergraph $H$ can appear in $G$ as subhypergraph at most $|V_H||V_H|\cdot \ell |V_H|\cdot \rho^*(H) \cdot m \cdot \rho^*(H)$ times. Given hypergraphs $H=(V_H,E_H)$ and $G=(V_G, E_G)$, if there are $t$ places in $G$ where $H$ appears as subhypergraph then obviously we cannot enumerate all of them in less than $t$ steps, however, there exists an algorithm that performs this enumeration in time that is polynomial in the upper bound $|V_H||V_H|\cdot \ell |V_H|\cdot \rho^*(H) \cdot m \cdot \rho^*(H)$.   We refer to this algorithm as {\em Find-Subhypergraph}.

\begin{theorem} \label{thm4} \cite{M08}  Let $H = (V_H, E_H)$ be a hypergraph with fractional cover number $\rho^*(H)$, and let $G(V_H, E_H)$ be a hypergraph where each edge has size at most $\ell$.  There is an algorithm that enumerates in time $|V_H|^{O(V_H)} \cdot \ell^{|V_H| \rho^*(H) + 1} \cdot |E_G|^{\rho^*(H) + 1} \cdot |V_G|^2$ every subset $V' \subseteq V_G$ where $H$ appears in $G$ as a subhypergraph. \end{theorem}
 
 
Given a {\em Consensus String with Outliers} instance with a set $S$ of $n$ length-$\ell$ strings and integer $n^*$, we define a minimal solution for this instance as a set $S^*_m$ and length-$\ell$ string $s_m$, where $\sum_{s_{i,m} \in S^*_m} d(s_m, s_{i,m})$ is minimal. 

\begin{theorem}\label{thm:d_fpt} Consensus String with Outliers can be solved in time $\delta^{O(\delta)} \cdot |\Sigma|^{\delta} \cdot n^9$ \end{theorem}

\begin{proof} Let $\{ S, k, d \}$ be an instance of {\em Consensus String With Outliers} with solution $S^*$ and $s$ denote the consensus string corresponding to $S^*$.  Clearly, $d(s, s_i^*) \leq \delta$ for at least one $s_i^* \in S^*$ and thus, if there exists a solution to a consensus string for $S^*$ then it can be found by considering all $s_0 \in S$ and checking if any string that has distance at most $\delta$ from $s_0$ is a consensus string for some subset of strings of $S$ of size $n^*$. Next, we show how to perform this analysis for one particular string $s_0 \in S$.  It follows that since there are at most $n$ possibilities for choosing $s_0$,  the running time of our algorithm for {\em Consensus String with Outliers} will be the running time of the following algorithm multiplied by a factor of $n$.  

Given $s_0 \in S$, we construct a hypergraph $G=(V,E)$, where $V = \{v_1, v_2, \ldots, v_{\ell}\}$ and the edge set describes the possible strings in the set of non-outlier strings of $S$. For each $s_i \in S$, there exists an edge $e_k \in E$ if and only if the symbol at the position $k$ of $s_0$ is not equal to the symbol at the position $k$ of $s_i$. Clearly, $G$ has at most $n$ edges.  Suppose $S^*$ is a solution to the original instance then we denote $H=(V_H, E_H)$ as the partial hypergraph in $G$ that contains the $n^*$ edges corresponding to the strings in $S^*$. 

Let $S^*_m$ and $s_m$ be a minimal solution to our original instance.  Denote $P$ as the set of positions where $s_m$ and $s_0$ differ and let $H_0$ be the subhypergraph of $H$ induced by $P$, i.e.~the vertex set of $H_0$ is equal to the vertices corresponding to the positions in $P$, and for each edge $e \in E$ there is an edge $E \cap P$ in $H_0$. Since $H_0$ is a subhypergraph of $H$ and $H$ is a partial hypergraph of $G$, it follows that $H_0$ appears in $G$ at $P$ as a subhypergraph.  The following proposition shows the fractional cover number of $H_0$ is at most 5/2 since the definition of a minimal solution to {\em Consensus Patterns} is identical to our definition of a minimal solution to {\em Consensus String with Outliers}.

\begin{proposition} \label{prop:frac_cover_no} \cite{M08}  Let $\{ S^*_m, s_m \}$ be a minimal solution to a Consensus Patterns instance, then the hypergraph $H_0^*$ corresponding to $\{ S^*_m, s_m \}$ has fractional cover number at most $5/2$. \end{proposition}

We can find all possible places $P$ by enumerating every suitable hypergraph $H_0$ and using Theorem \ref{thm4} to find all places where $H_0$ appears in $G$ as a subhypergraph.  In order to adequately bound on the running time indured by using the algorithm corresponding to Theorem \ref{thm4}, a bound on the size of the edges in $G$ is required.  It follows from the work of Marx \cite{M08} that we can remove every edge of size greater than $20\delta$ from $G$ (and $H$ respectively). Let $G^*$ (and $H^*$ respectively) be the resulting hypergraph and $H^*_0$ be the subhypergraph of $H^*$ induced by $P$.  Since $H^*_0$ is subhypergraph of $G^*$ and the fractional edge cover number can be bounded by a constant (Proposition \ref{prop:frac_cover_no}), we can find all the possible places $P$ by enumerating every hypergraph $H_0^*$ on $\delta$ vertices having fractional cover number at most $5/2$  and finding every place in $G^*$ where $H_0^*$ appears.  The following proposition demonstrates that we only need to consider hypergraphs that have $O(\log \delta)$ edges, further restricting the hypergraphs that need consideration. 

\begin{proposition} \label{prop:6_3} \cite{M08}  Let $\{ S^*_m, s_m \}$ be a minimal solution to a Consensus String with Outliers instance, and $H_0^*$ is the corresponding hypergraph, then it is possible to select $200 \log \delta$ edges of $H^*_0$ in such a way that if we delete all other edges, then the resulting hypergraph $H^{**}_0$ has fractional cover number at most 5. \end{proposition}

\begin{algorithm*}[h]
\caption{Consensus String with Outliers $d$-Parameterization Algorithm}
\label{alg:param_d}
\begin{algorithmic}
\STATE 1: For each string $s_0 \in S$: 
\STATE 2: \hspace{5mm} Construct the hypergraph $G^*$ on $\{1,2, \ldots, \ell\}$. 
\STATE 3: \hspace{5mm} For each hypergraph $H^{**}_0$ having $\leq \delta$ vertices and $\leq 200 \log \delta$ edges:
\STATE 4: \hspace{10mm} If every vertex of $H^{**}_0$ is covered by at least 1/5 part of the edges then: 
\STATE 5: \hspace{15mm} For every place $P$ where $H^{**}_0$ appears in $G^*$ as a subhypergraph:
\STATE 6: \hspace{20mm} For every string $s$ that differs from $s_0$ at the positions corresponding to $P$: 
\STATE 7: \hspace{25mm} Let $S^* \subset S$ of size $n^*$, where $d(s_0, s_i') \leq d(s_0, s_j)$, $\forall s_i' \in S^*$, $\forall s_j \in S / S^*$. 
\STATE 8: \hspace{25mm} If $d(s, s_i') \leq \delta$, for all $s_i' \in S^*$ then:
\STATE 9: \hspace{30mm} Return $s_0$ and $S^*$.
\STATE 10: Return ``no solution'' and halt.
\end{algorithmic}
\end{algorithm*}

There are $n$ possible choices for $s_0$ in the first step and the remainder of the algorithm checks whether there is a consensus string that differs from $s_0$ in at most $\delta$ positions. Constructing the hypergraph $G^*$ can be done in $O(\ell n)$ time. Since the aim is to find strings $s$ where $d(s_0, s) \leq \delta$, we can assume that $H^{**}_0$ has at most $\delta$ vertices; there are at most $2^{\delta \log \delta} = 2^{O(\delta)}$ unique hypergraphs with at most $\delta$ vertices and at most $200 \log \delta$ edges since there are at most $2^{\delta}$ possibilities for each edge.  Therefore, Step 3 enumerates through at most $O(2^{O(\delta \log \delta)})$ hypergraphs. The test in Step 4 is trivial. Step 5 is performed using the {\em Find-Subhypergraph} corresponding to Theorem \ref{thm4}.  It follows from the fact that the fractional cover number of $H_0^{**}$ is at most 5 and every edge of $G^*$ has size at most $20\delta$, that Step 5 takes $\delta^{O(\delta)} n^6\ell^2$ time. If $H_0^{**}$ appears at $P$ in $G^*$ as subhypergraph, then Step 6 considers at most $|\Sigma|^{\delta}$ possible strings and testing each string takes $O(\ell n)$ time.  Therefore, the total running time is $\delta^{O(\delta)} |\Sigma|^{\delta} n^9$. 

\end{proof}

\section{Conclusions and Future Work}

We presented the {\em Consensus String with Outliers} problem with the aim to model error correction of genomic data, and demonstrated that studying its parameterized complexity and approximability leads to surprising theoretical results.  We studied the complexity of {\em Consensus String with Outliers} with respect to different parameterizations, Table \ref{tab:results} summarizes these results.  Majority of these results are proved using standard parameterized reductions and hence, we leave them to the Appendix.  The most notable of these results demonstrates that {\em Consensus String with Outliers} parameterized by $\frac{d}{n-k}$ is FPT.

\begin{table}
\begin{center}
\begin{tabular}{@{\hspace{0.5cm}}l  @{\hspace{0.5cm}}c @{\hspace{0.5cm}}c }
	\hline
  Parameter(s) 		 		& $|\Sigma|$ is bounded 									& $|\Sigma|$ is unbounded  	\\
  \hline
  $\ell, d, n^*$ 			& FPT 																	& W[1]-hard \\  
  $\ell$  						& FPT 																	& W[1]-hard \\ 
  $n^*$ 						& W[1]-hard															& W[1]-hard \\  
  $k$ 							& W[1]-hard  														& W[1]-hard  \\ 
  $d$ 							& FPT 																	& W[1]-hard  \\ 
  \hline
\end{tabular}\end{center}
\caption{An overview of the fixed parameter tractability and intractability of the {\em Consensus String with Outliers}.}
\end{table}\label{tab:results}

Our results rule out the possibility of a $(1+\epsilon)$ approximation algorithm that has running time $O\left(f(1/\epsilon)n^{O(1)} \right)$, while our PTAS has running time $O \left(n^{1 / \epsilon^4} \right)$. Hence there is still a significant gap between known upper and lower bounds for the running time of approximation schemes for the problem. Obtaining tighter bounds warrants further investigation.  

Another problem that is FPT parameterized by objective function value, admits a PTAS but is not known to admit an EPTAS is the {\em Consensus Patterns} problem \cite{M08}, which seems to be closely related to {\em Consensus String with Outliers}. It is quite possible that our results on random walks, and hardness proofs could be useful to rule out an EPTAS for {\em Consensus Patterns}, which would answer an open problem given by Fellows et al.~\cite{FGN06}, and for other problems as well.

\newpage
\section{Appendix}   

We prove that when the alphabet size is unbounded {\em Consensus String With Outliers} is W[1]-hard for every combination of the parameters $\ell$, $d$, and $n^*$.  We define an instance of {\em Clique} by an undirected graph $G=(V, E)$ with a set $V=\{v_1, v_2, \ldots, v_n\}$ of $n$ vertices, a set $E$ of $m$ edges, and a positive integer $t$ denoting the size of the desired clique.  We generate a set $S$ of ${{t}\choose{2}} |E|$ strings such that $G$ has a clique of size $t$ if and only if there is a subset of $S$ of size ${{t}\choose{2}}$, denoted as $S^*$, where there exists a string $x$ such that $\sum_{\forall s_i \in S^*} d(s_i, x) \leq d$.  We let $\ell = t$ and $d = {t \choose 2}(t - 2)$.  We assume that $t > 2$ since $t \leq 1$ produces trivial cases.  

\begin{theorem} Consensus String with Outliers with an unbounded alphabet is W[1]-hard with respect to the parameters $\ell$, $d$, and $n^*$.\end{theorem}

\begin{proof}  We begin by describing the alphabet.  We assume $|\Sigma|$ can be infinite and we let $\Sigma$ be equal to the union of the following sets of symbols:
\begin{enumerate}
\item $\{v_i | \mbox{  for all } i = 1, \ldots, |V|\}$.  Hence, there exists one symbol representing each vertex in $G$.
\item $\{c_{i,j,m} | i = 1, \ldots, t; \, j = 1, \ldots, t; m = 1, \ldots, |E|\}$.  There exists an unique symbol for each ${t\choose 2} \cdot |E|$ strings produced for our reduction. 
\end{enumerate}
Hence, we have a total of $|V| + {t\choose 2} \cdot |E|$ number of symbols.
 
We construct a set of ${{t}\choose{2}} |E|$ strings $S = \{s_{1,1,1},  \ldots, s_{1,1,|E|}, s_{1,2,1}, \ldots, $ $s_{1,2,|E|}, \ldots, s_{t -1,t, |E|}\}$. Every string has length $t$ and will encode one edge of the input graph. There will be ${t \choose 2}$ corresponding for each edge, however, encode the edges in different positions.  For string $s_{i,j,m}$ we encode edge $e_m=  (v_r, v_s)$, where $1 \leq r < s \leq |V|$, but letting position $i$ equal to $v_r$ and position $j$ equal to $v_s$ and the remaining positions equal to $c_{i,j,m}$. Hence, a string is given by $$s_{i,j,m}\, := \, [ c_{i,j,m} ]^{i - 1} v_r [ c_{i,j,m} ]^{j - i - 1} v_s [ c_{i,j,m} ]^{m - j}.$$

To clarify our reduction, we give an example. Let $G = (V, E)$ be an undirected graph with $V = {v_1, v_2, v_3, v_4}$ and edges $E = \{(v_1, v_2), (v_1, v_3), (v_1, v_4), (v_2, v_3)\}$ and let our {\em Clique} instance have $G$ and $t = 3$.   Using $G$, we exhibit the above construction of ${t \choose 2} \cdot |E| = 12$ strings, which we denote as $S$.  We claim that there exists a clique of size 3 if and only if there exists a string $s^*$ of length $\ell = t = 3$ and subset $S^*$ of $S$ of size $3$ where $d(S^*, c(S^*)) \leq d$.  

First, we show that for a graph with a clique of size $t$, the above construction produces an instance of {\em Consensus String with Outliers} with a set $S^*$, consensus string $c(S^*)$ of length $\ell$ such that $d(S^*, c(S^*)) \leq d$.  Let the input graph have a clique of size $t$.  Let $v_{\alpha_1}, v_{\alpha_2}, \ldots, v_{\alpha_t}$ be the vertices in the clique $C$ of size $t$ and without loss of generality, assume $\alpha_1 < \alpha_2 < \ldots < \alpha_t$.  Then we claim that the there exists a subset of ${t\choose 2}$ vertices that have distance at exactly $t - 2$ from the string $s = v_{\alpha_1} v_{\alpha_2} \ldots v_{\alpha_t}$.   Consider the first edge of the clique $(v_{\alpha_1}, v_{\alpha_2})$ of the clique then it follows that the string $s_{11r} = v_{\alpha_1} v_{\alpha_2} [c_{11r}]^{t - 2}$, where edge $r$ has endpoints  $v_{\alpha_1}$ $v_{\alpha_2}$, is contained in the set of strings $\{s_{111}, s_{112}, \ldots, s_{11|E|} \}$. Clearly, $H(s_{11r}, s) = t - 2$. For each edge in $C$ we have we have a string in $S$ that has distance $t - 2$ from $s$ and our lemma follows from this construction.  

For the reverse direction, we need to prove that the existence a subset $S^*$ of size ${t \choose 2}$, where $d(S^*, c(S^*)) \leq {t \choose 2} (t - 2)$  implies the existence of a clique in $G$ with $t$ vertices.  Let $S^*$ be the subset of $S$ of size $t\choose 2$ such that $s$ has distance ${t \choose 2}(t - 2)$ from each string in $S^*$. Since $\ell = t$, $n^* = {t \choose 2}$, $d = {t \choose 2} (t -2)$ and he symbol $c_{i,j,m}$ occurs in only a single string in $S$ for all $i = 1, \ldots, t$, $ j = 1, \ldots, t$ and $m = 1, \ldots, |E|$, it follows from the Pigeonhole principle that the consensus string only contains symbols from the set $\{v_i | \mbox{  for all } i = 1, \ldots, |V|\}$. Without loss of generality assume the consensus string is equal to $v_{\alpha_1} v_{\alpha_2}\ldots v_{\alpha_t}$ for $\alpha_{v_1}, \alpha_{v_2}, \ldots, \alpha_{v_t} \in \{1, \ldots, |V|\}$.  Consider any pair $\alpha_i$, $\alpha_j$ for $1 \leq i < j \leq t$ and the set of strings $S_{i,j} = \{ s_{i,j,1}, s_{i,j,2}, \ldots, s_{i,j,|E|} \}$.  Recall that $S_{i,j}$ contains a string corresponding to each edge $e = (r, s)$ in $E$ which has $v_r$ at the $i$th position and $v_s$ at the $j$th position and $c_{i,j,m}$ at all remaining positions.  Therefore, we can only find a string in $S_{i,j}$ that has distance $t - 2$ from $s$ if $v_{\alpha_i}$ is at the $i$th position and $v_{\alpha_j}$ is at the $j$th position; and such a string exists if and only if there is an edge in $G$ connecting $v_{\alpha_i}$ to $v_{\alpha_j}$. Hence, the consensus string $s$ implies there exists an edge between any pair of vertices in $G$ in the set $\{v_{\alpha_1} v_{\alpha_2}\ldots v_{\alpha_t}\}$ and by definition the vertices form a clique. \hfill $\Box$  \end{proof}

Our main theorem follows directly from Lemma \ref{direction_one_lemma} and Lemma \ref{direction_two_lemma}.  We note that the hardness for the combination of all three parameters also implies the hardness for each subset of the three.

\setcounter{equation}{0}

\end{document}